\documentclass{article}
\usepackage[margin=1in,hmargin=1in]{geometry}
\usepackage{booktabs}
\usepackage{enumerate}
\usepackage{tikz}
\usepackage{float}
\usepackage{url}
\usepackage{graphicx}  
\usetikzlibrary{matrix}
\usepackage{natbib}
\bibpunct{(}{)}{;}{a}{}{;}

\usepackage{amsmath,amssymb,amsthm}
\newtheorem{theorem}{Theorem}
\newtheorem{lemma}{Lemma}

\newcommand{\dhat}[1]{ \widehat{\widehat{#1}} }

\title{Structural Break Detection in High-Dimensional Non-Stationary VAR models}
\begin{document}

\maketitle

\centerline{Abolfazl Safikhani \footnote[1]{as5012@columbia.edu} and Ali Shojaie \footnote[2]{ashojaie@uw.edu}}

\centerline{Columbia University and University of Washington }

\vspace{1cm}

\textbf{Abstract}
Assuming stationarity is unrealistic in many time series applications. A more realistic alternative is to allow for piecewise stationarity, where the model is allowed to change at given time points. In this article, the problem of detecting the change points in a high-dimensional piecewise vector autoregressive model (VAR) is considered. Reformulated the problem as a high-dimensional variable selection, a penalized least square estimation using total variation LASSO penalty is proposed for estimation of model parameters. It is shown that the developed method over-estimates the number of change points. A backward selection criterion is thus proposed in conjunction with the penalized least square estimator to tackle this issue. We prove that the proposed two-stage procedure consistently detects the number of change points and their locations. A block coordinate descent algorithm is developed for efficient computation of model parameters. The performance of the method is illustrated using several simulation scenarios.

\textbf{Keywords:} High-dimensional time series; Structural break; LASSO; Piecewise stationary.

\section{Introduction}\label{sec:intro}
Emerging applications in biology \citep{michailidis_2013autoregressive, smith2012future, fujita2007modeling, mukhopadhyay2006causality} and finance \citep{de_2008, fan_2011sparse} have sparked an interest in methods for analyzing high-dimensional time series. Recent work includes new regularized estimation procedures for vector autoregressive (VAR) models \citep{Basu_2015, matteson_2017}, high-dimensional generalized linear models \citep{hall_2016} and high-dimensional point processes \citep{HansenETAL_2015,  ChenWittenShojaie_2017}. These methods generalize the earlier work on methods for high-dimensional longitudinal data \citep{ShojaieMichailidis_2010, ShojaieBasuMichailidis_2012}, and handle the theoretical challenges of resulting from the temporal dependence among observations. 
Related methods have also focused on joint estimation of multiple time series \citep{QiuETAL_2016}, estimation of (inverse) covariance matrices \citep{xiao2012covariance, chen2013covariance, tank2015bayesian}, and estimation of high-dimensional systems of differential equations \citep{LuETAL_2011, ChenShojaieWitten_2016}. 

Despite considerable progress, both on computational and theoretical fronts, the vast majority of existing work on high-dimensional time series assumes that the underlying process is \emph{stationary}. However, multivariate time series observed in many modern applications are nonstationary. For instance, \citet{ClaridaGaliGertler_2000} show that the effect of inflation on interest rates varies across Federal Reserve regimes. Similarly, as pointed out by \citet{OmbaoVonSachsGuo_2005}, electroencephalograms (EEGs) recorded during an epileptic seizure display amplitudes and spectral distribution that vary over time. This nonstationarity is illustrated in Figure~\ref{fig_EEG_full}, which shows the EEG signals recorded at 18 EEG channels during an epileptic seizure from a patient diagnosed with left temporal lobe epilepsy \citep{OmbaoVonSachsGuo_2005}. 
The sampling rate in this data is 100~Hz and the total number of time points per EEG is $ T =32,768 $ over $\sim238$ seconds. Based on the neurologist's estimate, the seizure took place at $t = 185~s$. The plot of the EEGs also suggests that the magnitude and the variability of these signals change  simultaneously around that time. Assuming stationarity when analyzing such high-dimensional times series can severely bias estimation and inference procedures.

\begin{figure}[t]
\begin{center}
\includegraphics[width=0.6\linewidth]{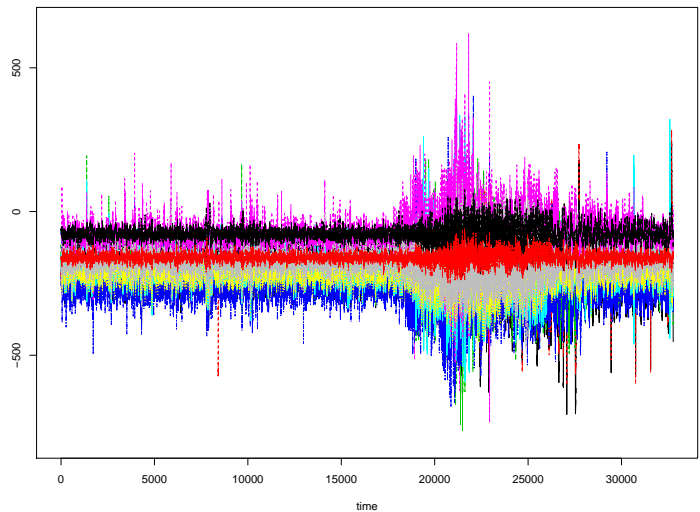}
\vspace{-0.5cm}
\caption{EEG signals from a patient diagnosed with left temporal lobe epilepsy. The data was recorded at 18 locations on the scalp during an epileptic seizure over 32,768 time points.}\label{fig_EEG_full}
\end{center}
\end{figure} 


Non-stationary VAR models have been primarily studied in univariate or low-dimensional settings. Existing approaches include models that fully parameterize the evolution of the transition matrices of time-varying VARs, or enforce a Bayesian prior on the structure of the time-dependence \citep{Primiceri_2005}. An alternative approach is to assume that the VAR process is \emph{locally stationary}; locally stationarity means that, in each small time interval, the process is well-approximated by a stationary one. This notion has been studied in low-dimensions by \citet{Dahlhaus_2012}; \citet{SatoMorettinETAL_2007} proposed a wavelet-based method for estimating the time-varying coefficients of the VAR model. 

Recently, \citet{DingQiuChen_2016} considered estimation of high-dimensional time-varying VARs by solving time-varying Yule-Walker equations based on kernelized estimates of variance and auto-covariance matrices. This approach is a significant step forward, and facilitates estimation of nonstationary VAR models in high dimensions. However, local stationarity may not be a suitable assumption in many applications. For instance, when analyzing EEG data from patients who suffer from epileptic seizure, it is expected that interactions among brain regions change before and after the occurrence of seizure. Assuming that the process can be locally approximated by a stationary one at the time of seizure may be unrealistic. A more natural assumption in such settings is that the process is \emph{piecewise stationary} --- that the process is stationary in each of (potentially many) regions, e.g., before and after seizure. 

Existing methods for analyzing piecewise stationary time series have primarily focused on univariate time series. For instance, \citet{Davis_2006}, \citet{Chan_2014} and \citet{Bai_1997} propose different approaches for identifying structural breakpoints at which the behavior of a univariate time series changes. 
By identifying structural breaks in mean and/or covariance structures over time, these approaches provide more flexible than those assuming stationarity. However, their extension to multivariate and high-dimensional VARs have not been explored. The only exception is the SLEX method of \citet{OmbaoVonSachsGuo_2005}, who analyzed the data from Figure~\ref{fig_EEG_full} and identified break points associated with seizure using a wavelet-based approach. However, to deal with the large number of time series, \citet{OmbaoVonSachsGuo_2005} apply a dimension reduction step. Thus, their method does not reveal mechanisms of interactions among brain regions, which is a key interest in understanding changes in brain function before, during and after seizure. In this paper we bridge this gap by developing a regularized estimation procedure for high-dimensional piecewise stationary VARs with possibly many break points. The proposed approach first identifies the number of break points. It then determines the location of the break points and provides consistent estimates of model parameters. Simulated and real data examples are used to support the theoretical findings of the paper, and illustrate the flexibility of the proposed approach in applications. 

The rest of this paper is organized as follows. In Section~\ref{sec:model}, we describe the piecewise stationary model and the key assumptions. We also present our estimation framework for detecting structural breaks in piecewise stationary VARs. The asymptotic properties of the proposed method are discussed in Section~\ref{sec:asymptotic}. In particular, we show that under reasonable assumptions the structural breaks in high-dimensional VAR models are consistency estimated. To this end, we first establishing the prediction consistency of the proposed method in Section~\ref{sec:predconsistency}. Results of simulation experiments are presented in Sections\ref{sec:sims}. In Section~\ref{sec:data} we illustrate the utility of the proposed method by applying it to identify structural break points in two multivariate time series. We conclude the paper with a discussion in Section~\ref{sec:disc}. Technical lemmas and proofs are collected in the Appendix.

\section{Model and Method}\label{sec:model}
A piecewise stationary VAR model can be viewed as a collection of separate VAR models concatenated at multiple break points over the time period of the observed time series. More specifically, suppose there exist $ m_0 $ break points $ 0 = t_0 < t_1 < \cdots < t_{m_0} < t_{m_0 + 1} = T + 1 $ such that 
\begin{equation}
y_t = \sum_{i=1}^{d} \Phi^{(i,j)} y_{t-i} + \varepsilon_t, \hspace{2cm} t_{j-1} \leq t < t_j, \hspace{1cm} j = 1, 2, ..., m_0 + 1,
\end{equation}
where $ y_t $ is a $ p \times 1 $ vector of observed time series at time $ t $, $ \Phi^{(i,j)}$'s are $ p \times p $ spares coefficient matrices of the VAR process, $ \varepsilon_t $ is a multivariate Gaussian white noise with covariance matrix $ \Sigma_\varepsilon $. 

Our  goal is to detect the break points $ t_j$'s together with estimates of the coefficient parameters $ \Phi^{(i,j)}$'s in the high-dimensional case where $ p \gg T $. To this end, we adopt the idea of change-point detection in \cite{Harchaoui_2010} and \cite{Chan_2014}, and extend it to the multivariate, high-dimensional setting. Specifically, our estimation procedure utilizes the following linear regression representation of the VAR process
\begin{equation}\label{regression}
\begin{pmatrix} y_d^\prime \\  y_{d+1}^\prime \\ \vdots  \\ y_T^\prime \end{pmatrix} = \begin{pmatrix} y_{d-1}^\prime & \ldots & y_0^\prime &  & 0 & & \ldots &  & 0 & \\  y_{d}^\prime & \ldots & y_1^\prime & y_{d}^\prime & \ldots & y_1^\prime  & & & 0 &\\  & \vdots & & & & & \ddots & &   \\ y_{T-1}^\prime & \ldots & y_{T-d}^\prime & y_{T-1}^\prime & \ldots & y_{T-d}^\prime & \ldots &  y_{T-1}^\prime & \ldots &  y_{T-d}^\prime \end{pmatrix} \begin{pmatrix} \theta_1^\prime \\  \theta_{2}^\prime \\ \vdots  \\ \theta_n^\prime \end{pmatrix} + \begin{pmatrix} \varepsilon_d^\prime \\  \varepsilon_{d+1}^\prime \\ \vdots  \\  \varepsilon_T^\prime \end{pmatrix},
\end{equation}
where $ n = T - d + 1 $, $ \Phi^{(.,j)} =  \begin{pmatrix} \Phi^{(1,j)} & \ldots & \Phi^{(d,j)} \end{pmatrix} \in \mathbb{R}^{p \times pd} $, $  \theta_1 = \Phi^{(.,1)} $ and 
\begin{equation}
\theta_i =  \left\{
	\begin{array}{ll}
		\Phi^{(.,i+1)} - \Phi^{(.,i)}, & \mbox{when } i = t_j  \,\, \mbox{for some} \, j  \\
		0, & \mbox{otherwise},
	\end{array}
\right.
\end{equation}
for $ i = 2, 3, ..., n $. 

Equation~\ref{regression} can be written in a compact form as 
$$ \mathcal{Y} = \mathcal{X} \theta (n) + \varepsilon (n), $$ 
or, in a vector form, as 
$$ Y = Z \Theta + E, $$ 
where $ Y = \mbox{vec}(\mathcal{Y}) $, $ Z = I_p \otimes \mathcal{X} $, and $ E = \mbox{vec}(\varepsilon(n)) $. Denoting $ q = n p^2 d $, $ Y \in \mathbb{R}^{np \times 1} $, $ Z \in \mathbb{R}^{np \times q} $, $ \Theta \in \mathbb{R}^{q \times 1} $, and $ E \in \mathbb{R}^{np \times 1} $. 
Note that in this parameterization, $ \widehat{\theta}_i \neq 0 $, $ i \geq 2 $ implies  a change in the VAR coefficients. Therefore, the structural break points $ t_j, j = 1, \ldots, m_0$ can be estimated as time points $ i \geq 2 $, where $ \widehat{\theta}_i \neq 0 $. 
To this end, the first step of our procedure consists of estimating the parameters $ \Theta $ using an $\ell_1$ penalized least squares regression. Formally, 
\begin{equation}\label{eq_estimation}
\widehat{\Theta} = \mbox{argmin}_{\Theta} \frac{1}{n} \| Y - Z \Theta \|_2^2 + \lambda_n \sum_{i=1}^{n} \| \theta_i \|_1 .
\end{equation}

The optimization problem in \eqref{eq_estimation} is convex and can be efficiently solved using a  block coordinate descent algorithm \citep{tseng2009coordinate}. This algorithm involves updating one of the $ \theta_i$'s at each iteration, until convergence. The KKT conditions of problem \eqref{eq_estimation}, presented in Lemma~\ref{lemma_KKT} of Appendix~A show that for fixed $ i = 1, 2, \ldots, n $, each update of $ \theta_i$ at iteration $ h+1 $  can be calculated as 
\begin{equation}\label{eq_bcd}
\theta_i^{\prime} (h+1) = { \left( \sum_{l=i}^{n} Y_{l-1} Y_{l-1}^\prime \right) }^{-1} S \left( \sum_{l=i}^{n} Y_{l-1} y_{l} - \sum_{j \ne i} \left( \sum_{l = \max(i,j)}^{n} Y_{l-1} Y_{l-1}^\prime \right) \theta_j^\prime (h) ; \lambda \right).
\end{equation}
Here, $ S(. ; \lambda) $ is the element-wise soft-thresholding function on all the components of the input matrix, which maps its input $ x $ to $ x - \lambda $ when $ x > \lambda $, $ x + \lambda $ when $ x < - \lambda $, and $ 0 $ when $ |x| \leq \lambda $. The iteration stops when $ \| \theta (h+1) - \theta (h) \|_\infty < \mbox{tolerance} $; we set $ \mbox{tolerance} = 10^{-3} $. Note that in this algorithm, the whole block of $ \theta_i$ with $ p^2 d $ elements is updated at once which reduces the computation time dramatically. Also, in each update of $ \theta_i$, the previous updated values of other blocks, i. e., other $ \theta_j$'s with $ j \ne i $ are used to speed up the convergence. 

\subsection{Refining the Initial Estimate}

Despite its convenience and computational efficiency, estimates from  \eqref{eq_estimation} do not correctly identify the structural break points in the piecewise VAR process. In particular, our theoretical analysis in the next section shows that the number of estimated break points from \eqref{eq_estimation}, i.e., the number of nonzero $ \widehat{\theta}_i \neq 0 $, $ i \geq 2 $, over--estimates the true number of break points. This is because the design matrix $ \mathcal{X} $ may not satisfy the restricted eigenvalue condition \citep{BickelETAL_2009} necessary for establishing consistent estimation of parameters. Instead, in the next section we first establish prediction consistency of the model from \eqref{eq_estimation}. We then show that consistent break point detection may be indeed achieved without requiring parameter estimation consistency. To this end, we first establish that if the number of change points $ m_0 $ is known, the estimator \eqref{eq_estimation} can consistently recover the break points   (Section~\ref{sec:known}). Using a more careful analysis, we then show that in  the case when $ m_0 $ is unknown, the penalized least squares \eqref{eq_estimation} identifies a larger set of \emph{candidate} break points. 

Denote the set of estimated change points from \eqref{eq_estimation} by 
$$
	\mathcal{A}_n = \left\lbrace i \geq 2 :  \widehat{\theta}_i \neq 0 \right\rbrace. 
$$ 
The total number of estimated change points is then the cardinality of the set $ \mathcal{A}_n $. Thus, $ \widehat{m} = | \mathcal{A}_n | $. Let $ \widehat{t}_1, 
\ldots, \widehat{t}_{\widehat{m}}  $ be the estimated break points. Then, the relationship between $ \widehat{\theta}_j $ and $ \widehat{\Phi}^{(.,j)}$ in each of the estimated segments can be seen as:
\begin{equation}\label{equation }
\widehat{\Phi}^{(.,1)} = \widehat{\theta}_1, \hspace{1cm} \mbox{and} \hspace{1cm} \widehat{\Phi}^{(.,j)} = \sum_{i=1}^{\widehat{t}_j} \widehat{\theta}_i, \hspace{1cm} j = 1, 2, ..., \widehat{m}. 
\end{equation}
Our results in Section~\ref{sec:unknown} below show that $ \hat m \geq m_0 $. These results also show that there exist $ m_0 $ points within $\mathcal{A}_n$ that are  `close' to the true break points. These result justify the second step of our estimation procedure described in the next section, which searches over the break points in $\mathcal{A}_n$ in order to identify an optimal set of break points. In fact, it is shown in Section~\ref{sec:consistency} that using an information criterion combining (a) regular least squares, (b) the $ L_1 $ norm of the estimated parameters, and (c) a term penalizing the number of break points, we are able to complete the search and correctly identify the number of segments in the model. Additional details about the second stage procedure are given in Section~\ref{sec:consistency}.

\section{Theoretical Analysis}\label{sec:asymptotic}

\subsection{Assumptions}\label{sec:assumptions}

To establish the asymptotic properties of the proposed estimator, we make the following assumptions. 

\begin{itemize}
\item[A1] For each fixed $ j =1, 2, ..., m_0+1 $, the process $ y_t^{(j)} = \sum_{i=1}^{d} \Phi^{(i,j)} y_{t-i}^{(j)} + \varepsilon_t $ is a stationary Gaussian time series. Denote the covariance matrices $ \Gamma_j (h) = \mbox{cov} \left( y_t^{(j)}, y_{t+h}^{(j)} \right) $ for $ t, h \in \mathbb{Z} $. Also, assume that the spectral density matrices $ f_j (\theta) = \frac{1}{2\pi} \sum_{l \in \mathbb{Z}} \Gamma_j (l) e^{-i l \theta} $, for $ \theta \in [-\pi, \pi] $ exist, and further $$ \mathcal{M}(f_j) = \mbox{ess sup}_{\theta \in [-\pi, \pi]} \Lambda_{\mbox{max}} (f_j(\theta)) < + \infty, $$ and $$ \textbf{m}(f_j) = \mbox{ess sup}_{\theta \in [-\pi, \pi]} \Lambda_{\mbox{min}} (f_j(\theta)) > 0, $$ where $ \Lambda_{\mbox{max}}(A) $ and $ \Lambda_{\mbox{min}}(A) $ are the largest and smallest eigenvalue of the symmetric or Hermitian matrix $ A $, respectively.
\item[A2] All the matrices $ \Phi^{(.,j)} $ are sparse. More specifically, denoting the number of nonzero elements in the $ i$-th row of $ \Phi^{(.,j)} $ by $ s_{ij} $, $ i = 1, 2, ..., p $ and $ j = 1, 2, ..., m_0 $, we have $ s_{ij} \ll p $ for all $ i, j $. Moreover, there exist positive constants $ v, M_\Phi > 0  $, and a large enough constant $ \nu^\prime > 0 $ such that, $$ \min_{1 \leq j \leq m_0} \frac{\max_{1 \leq i \leq p}  \left|\left| \Phi_i^{(.,j+1)} \right|\right|_2  }{\max_{1 \leq i \leq p}  \left|\left| \Phi_i^{(.,j)} \right|\right|_2}  \geq v^\prime, \min_{1 \leq j \leq m_0} \left|\left| \Phi^{(.,j+1)} - \Phi^{(.,j)}  \right|\right|_2  \geq v, \hspace{0.1cm} \mbox{and}  \hspace{0.1cm} \max_{1 \leq j \leq m_0+1} \left|\left| \Phi^{(.,j)}  \right|\right|_\infty  \leq M_\Phi. $$
Moreover, for each $ j = 1, 2, ..., m_0 + 1 $ and $ i = 1, ..., p $, define $ NZ_{ij} $ to be the set of all column indexes of $ \Phi_i^{(.,j)} $ at which there is a nonzero term. Also define $ NZ = \cup_{i,j} NZ_{ij} $, and further define $ s^\star = \max_{1 \leq i \leq p, 1 \leq j \leq m_0+1} |NZ_{ij}| $. Then, we have $ s^\star \sqrt{\frac{\log p}{n \gamma_n}} \rightarrow 0 $ as $ n \rightarrow \infty $.
\item[A3] There exists a positive sequence $ \gamma_n  $ vanishing such that $ \min_{1 \leq j \leq m_0+1} | t_j - t_{j-1} | / (n\gamma_n) \rightarrow + \infty  $, $ \gamma_n  / \left( s^\star  \lambda_n \right)  \rightarrow + \infty $, and $ \log(p)/(n \gamma_n) \rightarrow 0 $. 
\end{itemize}

Assumption~A1 helps us achieve appropriate probability bounds needed in the proofs. The second part of A1 will also be needed in the proof of consistency of the VAR parameters once the break points are detected. Assumption~A2 is a minimum distance-type requirement between the coefficients in different segments. The sequence $ \gamma_n $ is directly related to the detection rate of the break points $ t_j$'s. Assumption~A3 connects this rate to the tuning parameter chosen in the estimation procedure.

\subsection{Prediction Error Consistency}\label{sec:predconsistency}
As pointed out earlier, and discussed in \cite{Chan_2014} and \cite{Harchaoui_2010}, the design matrix of the linear regression formulation of the piecewise VAR model may not satisfy the restricted eigenvalue condition needed for parameter estimation consistency \citep{BickelETAL_2009}.  Thus, as a first step in establishing the consistency of the proposed procedure, in this section we establish the prediction error consistency of LASSO estimator from \eqref{eq_estimation}.

\begin{theorem}\label{thm_pred_error}
Suppose A1 and A2 hold. Choose $ \lambda_n =  2 C \sqrt{\frac{\log(n) + 2\log(p) + \log(d)}{n}} $ for some $ C  > 0 $. Also, assume $ m_0 \leq m_n $ with $ m_n = o \left( \lambda_n^{-1} \right) $. Then, with high probability approaching to 1 as $ n $ goes to $ + \infty $, 
\begin{equation}\label{eq:predconsistency}
\frac{1}{n} \left|\left| Z \left( \widehat{\Theta} - \Theta  \right)  \right|\right|_2^2 \leq 4 C m_n   \max_{1 \leq j \leq m_0+1} \left\lbrace \sum_{i=1}^{p} \left( s_{ij} + s_{i(j-1)}  \right) \right\rbrace  M_\Phi \sqrt{\frac{\log(n) + 2\log(p) + \log(d)}{n}}.
\end{equation}
\end{theorem}

Theorem~\ref{thm_pred_error} is proved in Appendix~B. Note that this theorem imposes an upper bound on the model sparsity, as the right hand side of \eqref{eq:predconsistency} must go to zero as $ n \rightarrow \infty $. In Section~\ref{sec:consistency}, we specify the limit on the sparsity needed for consistent identification of structural break points. 

\subsection{The Case of Known $ m_0 $}\label{sec:known}

In this section, we study a simplified version of the problem, by assuming that the true number of change points are known. In this case, the task reduces to locating the break points. We obtain the following result for this simplified problem. 

\begin{theorem}\label{thm_known_m}
Suppose A1, A2, and A3 hold. If $ m_0 $ is known and $ |  \mathcal{A}_n| = m_0 $, then
$$  
\mathbb{P} \left( \max_{1 \leq i \leq m_0} | \widehat{t}_i - t_i | \leq n \gamma_n \right) \rightarrow 1, \hspace{1cm} \mbox{as} \,\,\, n \rightarrow + \infty.  
$$
\end{theorem}

Theorem~\ref{thm_known_m} is proved in Appendix~B. In this theorem, the rate of consistency for this problem is $ n \gamma_n $, which can be chosen as small as possible assuming that conditions A2 and A3 hold. This is achieved by examining the KKT condition for the optimization problem \eqref{eq_estimation}, stated in Lemma~\ref{lemma_KKT} and using probability bounds in Lemma~\ref{lemma_bound}; these lemmas are given in the Appendix~A.
It is worth noting that $ \gamma_n $ also depends on the minimum distance between consecutive true break points, as well as the number of time series, $ p $. When $ m_0 $ is finite, one can choose $ \gamma_n = {(\log n \log p)}/{n} $ or $ \gamma_n = {(\log\log n \log p)}/{n} $. This means that the convergence rate for estimating the relative locations of the break points, i.e., $ t_i/T $ using $ \widehat{t}_i/T $ could be as low as $ {(\log\log n \log p)}/{n} $. In the univariate case, \cite{Chan_2014} showed a convergence of order $ {(\log n)}/{n} $. The rate found here is larger than the univariate case by an order less than $ \log p $ which is due to the growing number of time series. This logarithmic factor captures the additional difficulty in estimating the structural break points in high-dimensional settings. 

\subsection{The Case of Unknown $ m_0 $}\label{sec:unknown}
We now turn to the more general case of unknown $ m_0 $. Our next result shows that the number of selected change points $ \widehat{m} $ based on the estimation procedure \eqref{eq_estimation} will be at least as large as the true number $ m_0 $. Moreover, each true change point will have at least one estimated point in its $ n \gamma_n$-radius neighborhood. 

Before stating the theorem, we need some additional notations. Let $ \mathcal{A} = \lbrace t_1, t_2, ..., t_{m_0} \rbrace $ be the set of true change points. Following  \cite{Boysen_2009} and \cite{Chan_2014}, define the Hausdorff distance between two sets as  
$$ 
d_H (A, B) = \max_{b \in B} \min_{a \in A} |b - a|. 
$$
We obtain the following results. 

\begin{theorem}\label{thm_Hausdorff}
Suppose A1, A2, and A3 hold. Then as $ n \rightarrow + \infty $,
$$ \mathbb{P} \left(  | \mathcal{A}_n | \geq m_0  \right) \rightarrow 1,$$
and
$$  \mathbb{P} \left(  d_H \left( \mathcal{A}_n, \mathcal{A}  \right)  \leq n \gamma_n \right) \rightarrow 1. $$
\end{theorem}

The second part of Theorem~\ref{thm_Hausdorff} shows that even though we  select more points than needed, there exists a subset of the estimated points with size $ m_0 $, which estimates the true break points at the same rate as if $ m_0 $ was known. This result motivates the second stage of our estimation procedure, discussed in the next section, which removes the additional estimated break points. 

\subsection{Consistent Estimation of Structural Breaks}\label{sec:consistency}
Theorem~\ref{thm_Hausdorff} shows that the penalized estimation procedure \eqref{eq_estimation} over-estimates the number of change points. A second stage screening is thus needed to consistently find the true number of change points. Our proposal, presented next, is a modification of the screening procedure of \cite{Chan_2014}. The basic idea is to develop an \emph{information criterion} based on a new penalized least squares estimation procedure, in order to screen the candidate break points found in the first estimation stage. Formally, for a fixed $ m $ and estimated change points $ s_1, ..., s_m $, we form the following linear regression:

\begin{equation}\label{eq_regression_second}
\begin{pmatrix} y_d^\prime \\  y_{d+1}^\prime \\ \vdots  \\ y_T^\prime \end{pmatrix} = \begin{pmatrix} Y_{d-1}^\prime  \\ \vdots & 0 & \ldots & 0 \\ Y_{s_1-1}^\prime \\ & Y_{s_1}^\prime \\ 0 & \vdots & \ldots & 0 \\ & Y_{s_2-1}^\prime \\ &&& \\ \vdots & \vdots & \ddots & \vdots \\ &&& \\  &&& Y_{s_m}^\prime \\ 0 & 0 && \vdots \\ &&& Y_{T}^\prime \end{pmatrix} \begin{pmatrix} \theta_1^\prime \\  \theta_{2}^\prime \\ \vdots  \\ \theta_{m+1}^\prime \end{pmatrix} + \begin{pmatrix} \varepsilon_d^\prime \\  \varepsilon_{d+1}^\prime \\ \vdots  \\  \varepsilon_T^\prime \end{pmatrix}.
\end{equation}
This regression can be written compactly as 
$$ 
\mathcal{Y} = \mathcal{X}_{s_1, ..., s_m} \theta_{s_1, ..., s_m} + \varepsilon (n), 
$$
where $ \mathcal{X}_{s_1, ..., s_m} \in \mathbb{R}^{n \times q_m} $,  $ \theta_{s_1, ..., s_m} = \left( \theta_{(1,s_1)}^\prime, \theta_{(s_1,s_2)}^\prime, ..., \theta_{(s_m,T)}^\prime \right)^\prime \in \mathbb{R}^{q_m \times p} $, with $ q_m = (m+1)pd $. We estimate $ \theta_{s_1, ..., s_m} $ using the following LASSO regression:
\begin{equation}\label{eq_estimation_second}
\widehat{\theta}_{s_1, ..., s_m} = \mbox{argmin}_{\theta} || \mathcal{Y} - \mathcal{X}_{s_1, ..., s_m} \theta ||_F^2 + n \,  \eta_n \sum_{i=1}^{m+1} || \theta_i ||_1,
\end{equation}
with tuning parameter $ \eta_n $. 

Define
\begin{equation}\label{ddd}
L_n(s_1, s_2, ..., s_m; \eta_n) =  || \mathcal{Y} - \mathcal{X}_{s_1, ..., s_m} \widehat{\theta}_{s_1, ..., s_m} ||_F^2 + n \, \eta_n \sum_{i=1}^{m+1} || \widehat{\theta}_{(s_{i-1},s_i)} ||_1,
\end{equation}
with $ s_0 = d $ and $ s_{m+1} = T $. Then, for a suitably chosen sequence $ \omega_n $, specified in Assumption~A4 below, consider the following information criterion:
$$ 
\mathrm{IC} (s_1, ..., s_m; \eta_n) =  L_n(s_1, s_2, ..., s_m; \eta_n)  + m \omega_n. 
$$

The second stage of our procedure selects a subset of $\dhat{m}$ break points by solving the problem 
\begin{equation}\label{eq_selection}
( \dhat{m}, \dhat{t}_1, ..., \dhat{t}_{\dhat{m}}  ) =  \mbox{argmin}_{0 \leq m \leq |\mathcal{A}_n|, \, \textbf{s} = (s_1, ..., s_m) \in \mathcal{A}_n } \mathrm{IC}(\textbf{s}; \eta_n).
\end{equation}


To establish the consistency of the proposed two-state selection procedure \eqref{eq_selection}, we need an additional assumption. 
\begin{itemize}
\item[A4] Let $ d_n^\star = \sum_{j=1}^{m_0+1} \sum_{i=1}^{p} s_{ij} $ be the total sparsity of the model. Then, $ m_0 n \gamma_n d_n^\star /\omega_n \rightarrow 0 $, and $ \min_{1 \leq j \leq m_0+1} | t_j - t_{j-1} | / (m_0 \omega_n) \rightarrow + \infty $. Also, either (a) $ m_0 \sqrt{\frac{\log p}{n \gamma_n}} = o (1) $ and $ \eta_n = \gamma_n $ or (b) $ m_0 \sqrt{\frac{\log p}{n \gamma_n}} = O (1) $ and $ \eta_n = C \gamma_n $ for some large enough positive constant $ C > 0 $. 
\end{itemize}

We can now state our main consistency result. 
\begin{theorem}\label{thm_selection}
Suppose A1, A2, A3, and A4 hold. Then, as $ n \rightarrow + \infty $, the minimizer  $ ( \dhat{m}, \dhat{t}_1, ..., \dhat{t}_{\dhat{m}}  ) $ of \eqref{eq_selection} satisfies 
$$ 
\mathbb{P} \left( \dhat{m} = m_0  \right) \rightarrow 1.
$$
Moreover, there exists a positive constant $ B > 0  $ such that 
$$
\mathbb{P} \left( \max_{1 \leq i \leq m_0} | \dhat{t}_i - t_i | \leq B n \gamma_n d_n^\star \right) \rightarrow 1.  
$$
\end{theorem}

The proof of the theorem, given in Appendix~B relies heavily on the result presented in Lemma~\ref{lemma_selection}, which is stated and derived in Appendix~A.  

\textbf{Remark 1}. For the case when $ m_0 $ is finite, the rates can be set to $ \gamma_n = {(\log n \log p)}/{n} $, $ \lambda_n = o \left( {(\log n \log p)}/{n p} \right) $, $ \eta_n = \gamma_n $, and $ \omega_n = (\log n \log p)^{1+v} $ for some positive $ v > 0 $. For these rates, the model can have total sparsity $ d_n^\star = o \left( (\log n \log p)^{v}  \right) $. 

\textbf{Remark 2}. The proposed two-stage procedure can be also applied to low-dimensional time series. For example, with $ p $ as low as $ p = c n^a $ for positive constants $ c, a $, the probability bounds derived in Lemma~\ref{lemma_bound} would be strong enough to get the desired consistency results shown for the high-dimensional case.

%
%

\section{Simulations}\label{sec:sims}

In this section, the performance of the proposed two stage model will be evaluated under different simulation scenarios. In all scenarios, 100 data sets are randomly generated with $ T = 300 $, $ p = 20 $, $ d = 1 $, $ m_0 = 2 $. All time series have mean zero, and $ \Sigma_\varepsilon = 0.01 I_T $. We consider three different scenarios.

\begin{enumerate}
\item \emph{Simple $ \Phi $ and break points close to the center}. In the first scenario, the autoregressive coefficients are chosen to have the same structure but different values as displayed in Figure~\ref{fig_phi}. In this scenario, $ t_1 = 100 $ and $ t_2 = 200 $, which means the break points are not close to the boundaries. 

Figure~\ref{fig_step_1} shows the selected break points in one out of 100 simulated data sets. As expected from Theorem~\ref{thm_Hausdorff}, more than 2 change points are detected using the first stage estimator. However, there are always points selected in a small neighborhood of true change points. The second screening stage eliminates the extra candidate points leaving with only two closest points to the true change points. Figure~\ref{fig_step_2} shows the final selected points in all the 100 simulation runs. The mean and standard devision of locations of selected points, relative to the sample size $ T $, are shown in Table~\ref{table_sim_1}. (More specifically, the mean and standard deviation of $ \dhat{t}_1 / T $ and $ \dhat{t}_2 / T $ are reported in the table.) It can be seen from the results that the two stage procedure accurately detects both the number of break points, as well as their locations.

\item \emph{Simple $ \Phi $ and break points close to the boundaries}. Here, $ t_1 = 30 $ and $ t_2 = 250 $. The final selected points are shown in Figure~\ref{fig_sim_3}, and mean and standard deviation of the location of selected points, relative to the sample size $ T $, are shown in Table~\ref{table_sim_3}. Compared to scenario 1, when when break points are closer to the boundaries, the estimated locations are less accurate. The results also show that some of the break points may not be correctly detected in this setting. 

\item \emph{Randomly structured $ \Phi $ and break points close to the center}. As in scenario 1, in this case we set $ t_1 = 100 $ and $ t_2 = 200 $. However, the coefficients are chosen to be randomly structured. As a result, detecting break points is more challenging in this setting. 
The autoregressive coefficients for this scenario are displayed in Figure~\ref{fig_phi_2}. 

The selected break points in this scenario are shown in Figure~\ref{fig_step_sim_2}, and the mean and standard deviation of locations of the selected break points, relative to the sample size $ T $, are shown in Table~\ref{table_sim_2}. 
The results suggest that this setting---with randomly structured $ \Phi$'s---is the most difficult scenario. In fact, the identification of the number of change points in this setting, as measured by the selection rate of the break points, is the worst among the simulations considered---the detection rate drops to $ 92\% $ compared to $100\%$ in scenario 1. Further, the standard deviation of the selected break point locations are considerably larger. The inferior performance of the proposed method in this scenario could be due to the fact that the $ L_2 $ distance between the consecutive autoregressive coefficients are less than the previous two cases. This would make it harder to identify the exact location of the break points.
\end{enumerate}

\begin{figure}[H]
\begin{center}
\includegraphics[height=0.3\textheight]{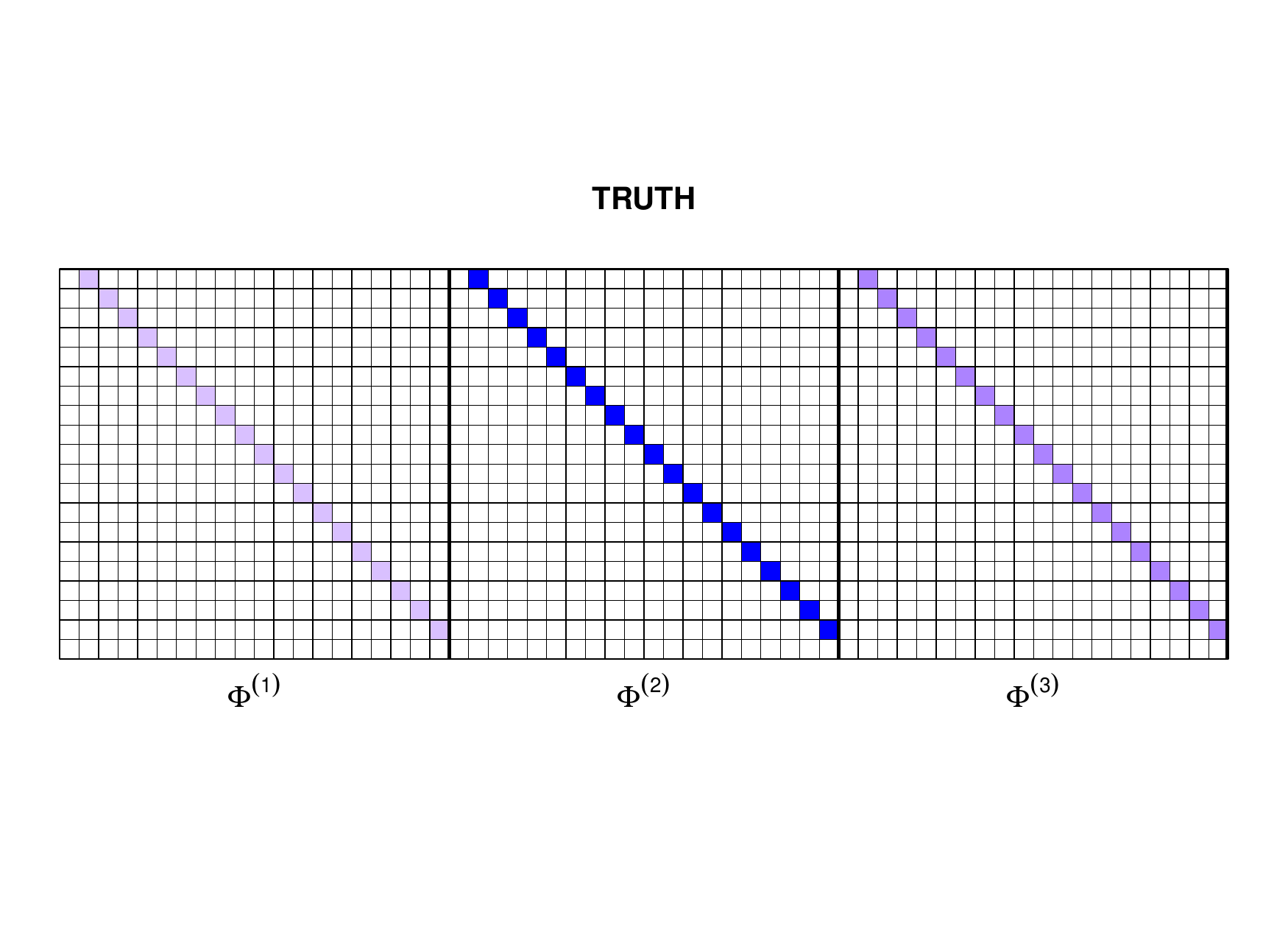}
\vspace{-2.0cm}
\caption{True autoregressive coefficients for the three segments used in the simulation scenario 1. }
\label{fig_phi}
\end{center}
\end{figure}

\vspace{-0.75cm}

\begin{figure}[H]
\begin{center}
\includegraphics[height=0.3\textheight]{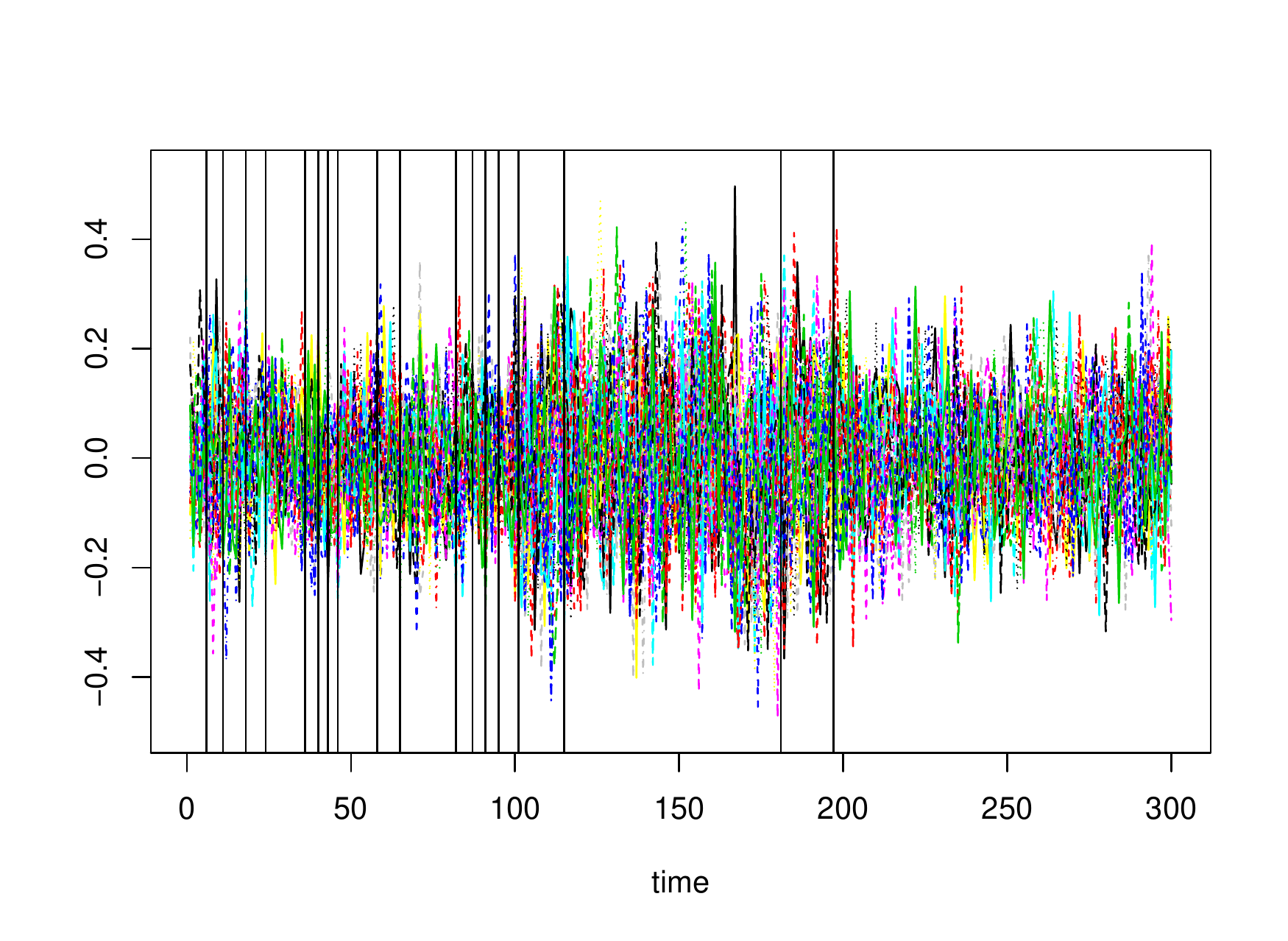}
\vspace{-0.5cm}
\caption{ Estimated break points on the first stage for one of the runs in simulation scenario 1: Close to 18 points are selected in the first stage.}
\label{fig_step_1}
\end{center}
\end{figure}

\vspace{-0.5cm}

\begin{figure}[H]
\begin{center}
\includegraphics[height=0.3\textheight]{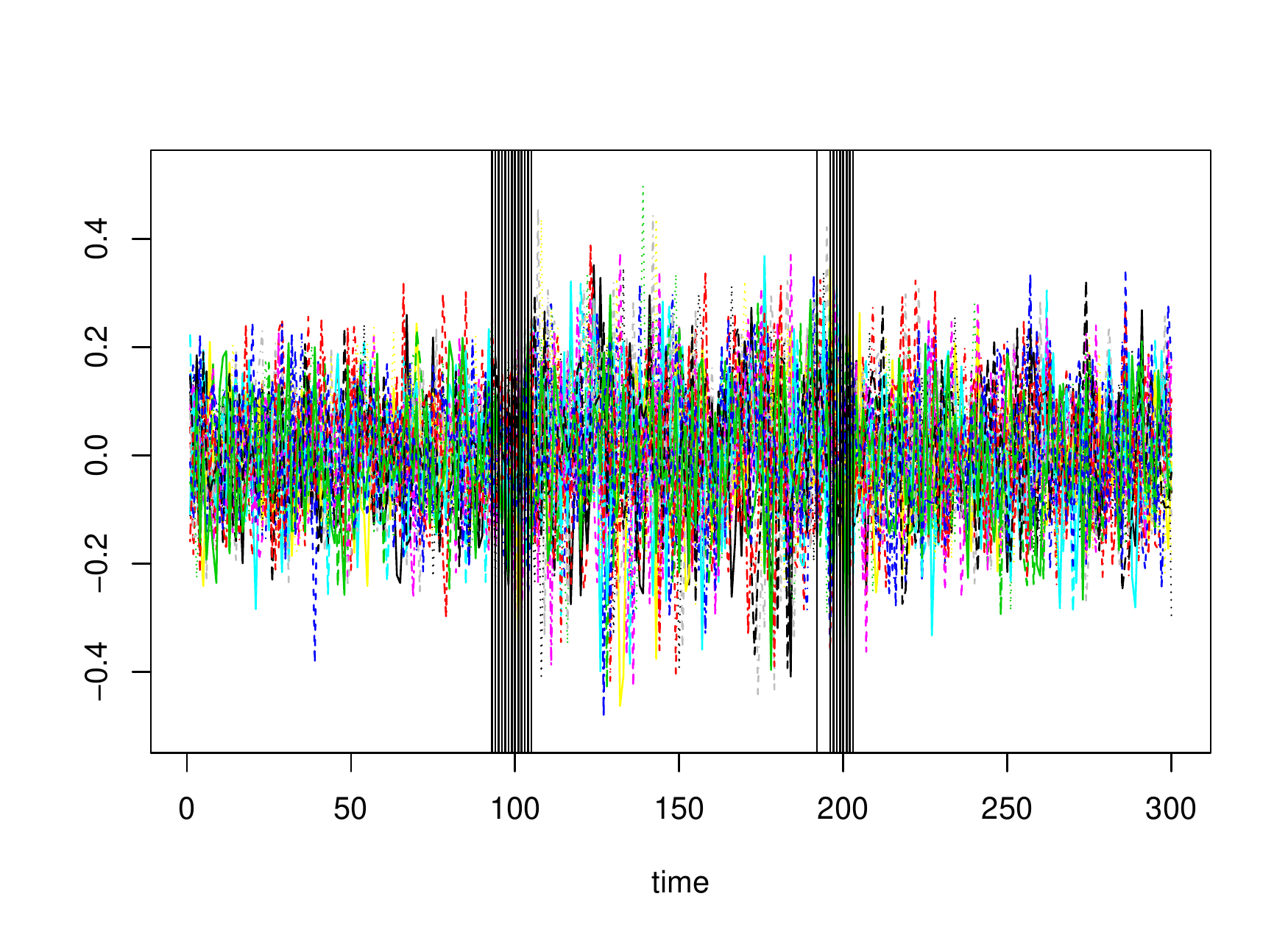}
\vspace{-0.5cm}
\caption{Final selected points for all the 100 runs from simulation scenario 1.}\label{fig_step_2}
\end{center}
\end{figure}

\begin{table}[ht]
\centering
\begin{tabular}{lllll} 
  \hline
break points & truth & mean & std & selection rate \\ 
  \hline
    \hline
  1 & 0.3333 & 0.3315 & 0.0074 & 1 \\ 
  2 & 0.6667 & 0.6632 & 0.0044 & 1 \\  
   \hline
\end{tabular}
\caption{Results of simulation scenario 1. The table shows mean and standard deviation of locations of selected break points, as well as the percentage of simulation runs where break points are correctly detected.}\label{table_sim_1}
\end{table}

\begin{figure}[H]
\begin{center}
\includegraphics[height=0.3\textheight]{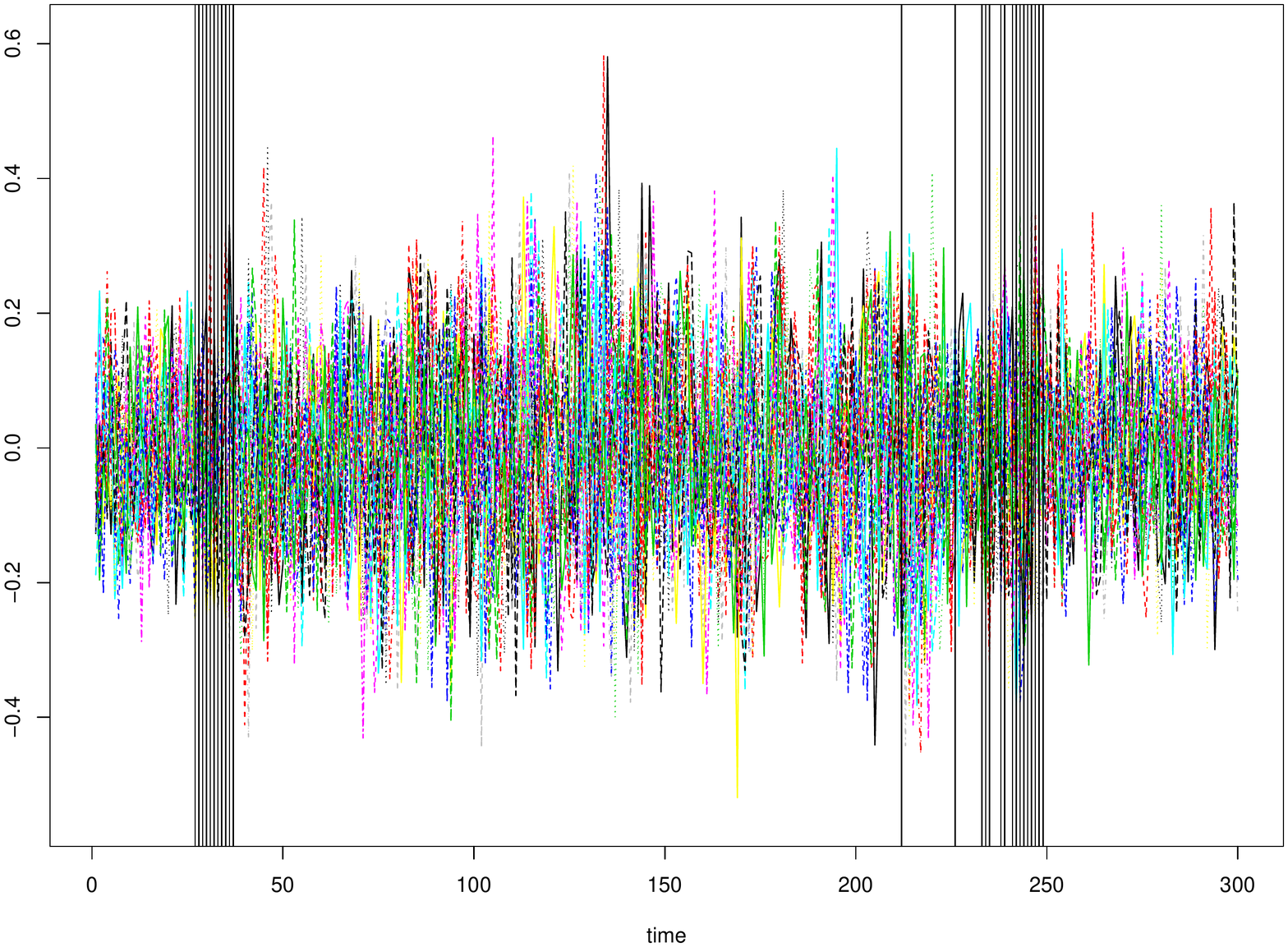}
\vspace{-0.5cm}
\caption{Final selected points for all the 100 runs from simulation scenario 2.}\label{fig_sim_3}
\end{center}
\end{figure}


\begin{table}[ht]
\centering
\begin{tabular}{lllll}
  \hline
break points & truth & mean & std & selection rate \\ 
  \hline
    \hline
  1 & 0.1 & 0.101 & 0.0082 & 0.98 \\ 
  2 & 0.8333 & 0.8134 & 0.0226 & 1 \\  
   \hline
\end{tabular}
\caption{Results of simulation scenario 2. The table shows mean and standard deviation of locations of selected break points, as well as the percentage of simulation runs where break points are correctly detected.}\label{table_sim_3}
\end{table}

\vspace{-0.5cm}

\begin{figure}[H]
\begin{center}
\includegraphics[height=0.3\textheight]{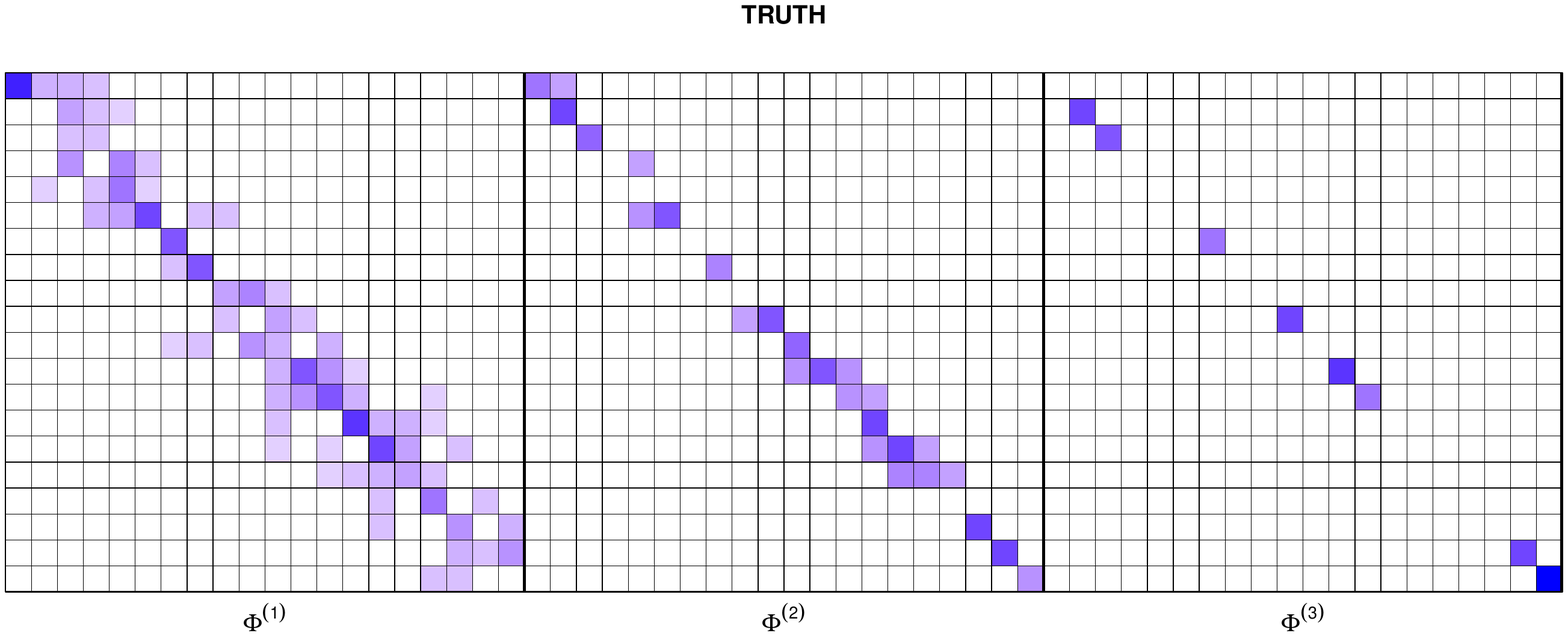}
\vspace{-1.5cm}
\caption{True autoregressive coefficients for the three segments used in the simulation scenario 3.}\label{fig_phi_2}
\end{center}
\end{figure}

\begin{figure}[H]
\begin{center}
\includegraphics[height=0.3\textheight]{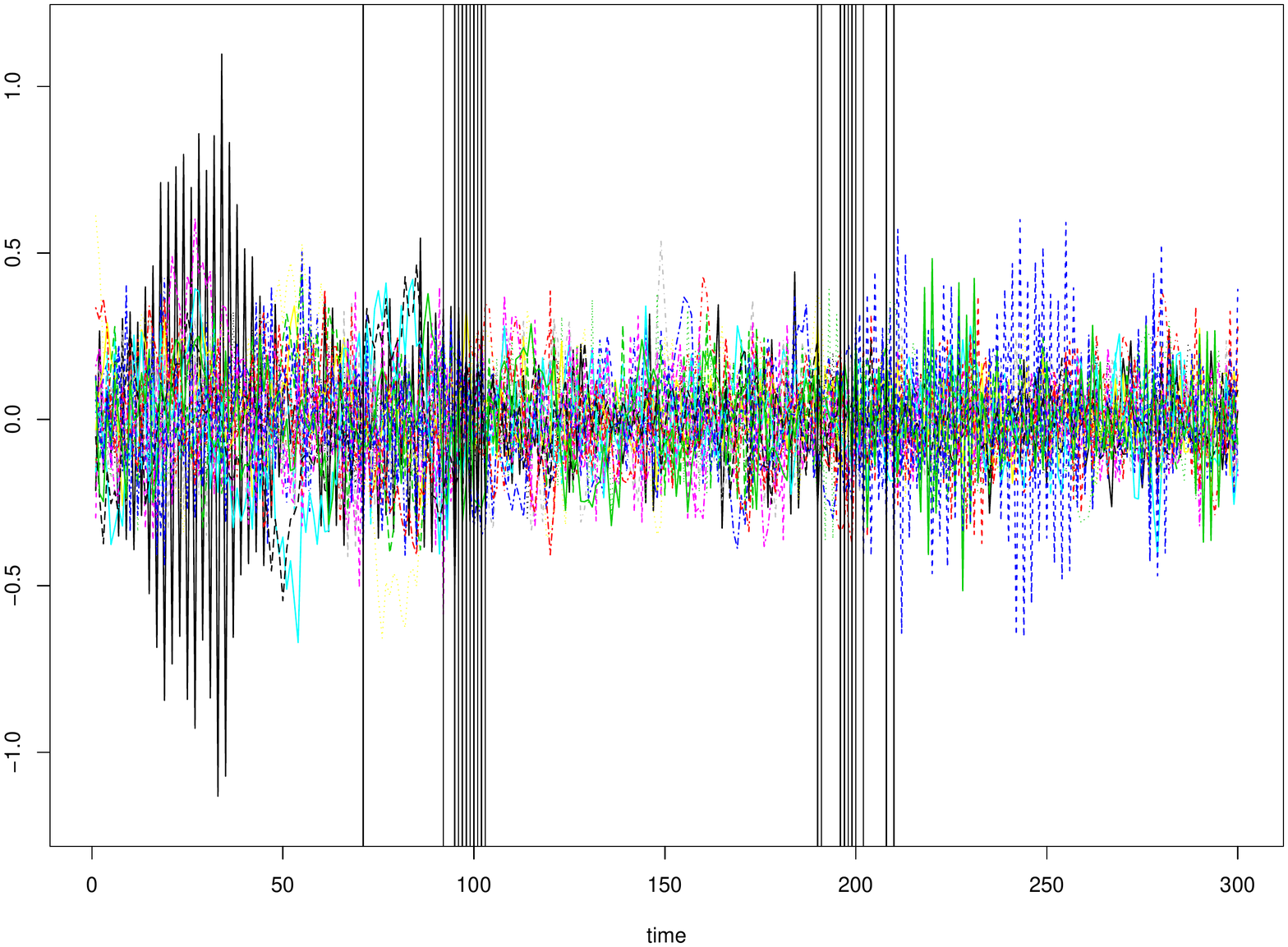}
\caption{Final selected points for all the 100 runs from simulation scenario 3.}\label{fig_step_sim_2}
\end{center}
\end{figure}

\begin{table}[ht]
\centering
\begin{tabular}{lllll}
  \hline
break points & truth & mean & std & selection rate \\ 
  \hline
    \hline
  1 & 0.3333 & 0.3282 & 0.0153 & 0.92 \\ 
  2 & 0.6667 & 0.6601 & 0.01 & 0.98 \\  
   \hline
\end{tabular}
\caption{Results of simulation scenario 3. The table shows mean and standard deviation of locations of selected break points, as well as the percentage of simulation runs where break points are correctly detected.}\label{table_sim_2}
\end{table}

\section{Real Data Applications}\label{sec:data}
In this section, we apply the proposed model to two real data sets in order to illustrate its performance in detecting break points in different settings. 

\subsection{EEG Data}\label{sec:EEG}
The data considered in this application consists of electroencephalogram (EEG) signals recorded at 18 locations on the scalp of a patient diagnosed with left temporal lobe epilepsy during an epileptic seizure. The sampling rate is 100 Hz and the total number of time points per EEG is $ T =32,768 $ over 238 seconds. The time series for all 18 EEG channels are shown in Figure~\ref{fig_EEG_full}. The seizure was estimated to take place at $ t = 185 s $. Examining the EEG plots, it can be seen that the magnitude and the volatility of signals change simultaneously around that time. 

To speed up the computations in this analysis, we selected one observation per second and reduced the total time points to $ T = 328 $. The EEG from a specific channel (P3) was previously used in \cite{Davis_2006} and \cite{Chan_2014}. Table~\ref{table_EEG} shows the location of the selected break points using the Auto--PARM method of \cite{Davis_2006}, the two-stage procedure of  \cite{Chan_2014} based on data from channel P3, and our proposed multivariate method. Our method correctly detects a break point at at $ t = 186 $, which is close to the seizure time identified by neurologists. The majority of other  selected break points by our method are close to the break points detected by the two univariate approaches. 

\begin{figure}[H]
\begin{center}
\includegraphics[width=0.6\linewidth]{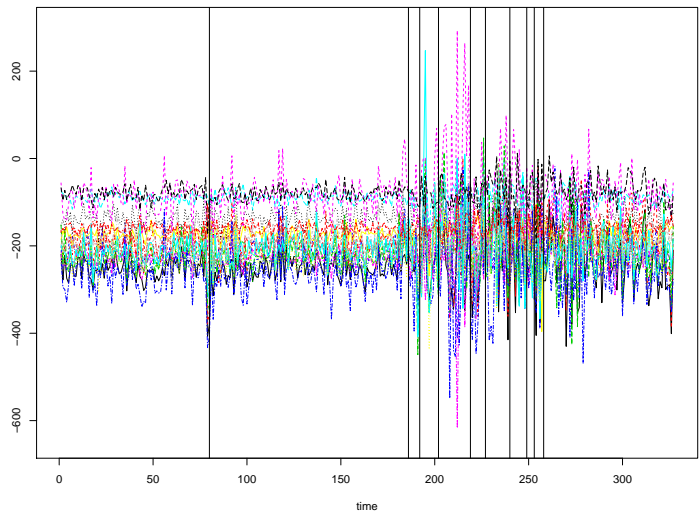}
\vspace{-0.5cm}
\caption{EEG data over 328 seconds with the 10 selected break points.}\label{fig_EEG_selected}
\end{center}
\end{figure}

\begin{table}[ht]
\centering
\begin{tabular}{llllllllllll}
\hline
Methods & 1 & 2 & 3 & 4 & 5 & 6 & 7 & 8 & 9 & 10 & 11 \\ 
  \hline
  \hline
  Auto--PARM & 186 & 190 & 206 & 221 & 233 & 249 & 262 & 275 & 306 & 308 & 326 \\ 
   Chan (2014) & 184 & 206 & 220 & 234 & 255 & 277 & 306 & 325 & -- & -- & -- \\ 
  Our method & 80 & 186 & 192 & 202 & 219 & 227 & 240 & 249 & 253 & 258 & --  \\ 
   \hline
\end{tabular}
\caption{Location of break points detected in the EEG data using three different methods. The locations are rounded to the closest integer.}\label{table_EEG}
\end{table}

\subsection{Yellow Cab Demand in NYC}\label{sec:taxi}

As a second example, we apply our method to the yellow cab demand data in  NYC. Here, the number of yellow cab pickups are aggregated spatially over the zipcodes and temporally over 15 minute intervals during April 16th, 2014. We only consider the zipcodes with more than 50 cab calls to obtain a better approximation using normal distribution. This results in 39 time series for zipcodes observed over 96 time points. To identify structural break points, we consider a differenced version of the data to remove first order non-stationarities. Table~\ref{table_Taxi} shows the 10 break points detected for this data; the differenced time series and the detected break points are also shown in Figure~\ref{fig_Taxi_selected}. 

Based on data from New York City metro (MTA), morning rush hour traffic in the city occurs between 6:30 AM and 9:30 AM, whereas the afternoon rush hour starts from 3:30 PM. Interestingly, among the selected break points, there are very  close to the rush hour start/end dates during a typical day. Specifically, the selected break points at 7 AM, 10 AM, 3:30 PM, and 6 PM are close to rush hour periods in NYC. These results suggest that the covariance structure of cab demands between the zipcodes in NYC may significantly change before and after the rush hour periods. 

\begin{figure}[H]
\begin{center}
\includegraphics[width=0.6\linewidth]{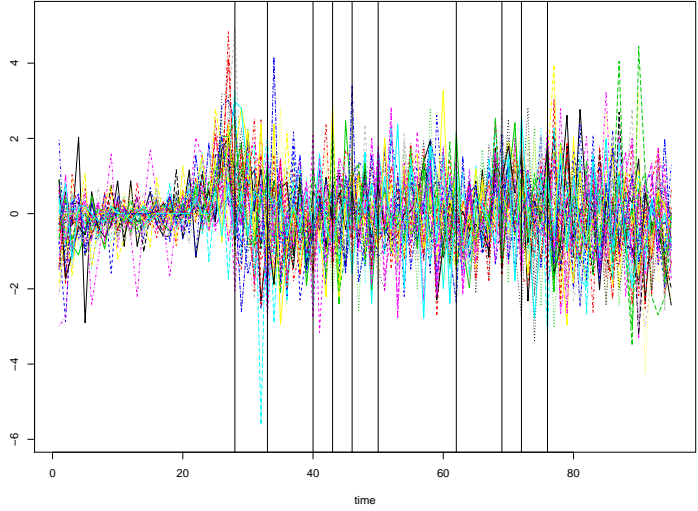}
\vspace{-0.5cm}
\caption{Plot of the NYC Yellow Cab Demand differenced time series from 39 different zipcodes over a single day with 96 time points; the 10 selected break points by the proposed method are shown as vertical lines.}\label{fig_Taxi_selected}
\end{center}
\end{figure}

\begin{table}[ht]
\centering
\begin{tabular}{llllllllllll}
\hline
  & 1 & 2 & 3 & 4 & 5 & 6 & 7 & 8 & 9 & 10  \\ 
  \hline
  \hline
  Our method & 7am & 8:15am & 10am & 10:45am & 11:30am & 12:30am & 3:30pm & 5:15pm & 6pm & 7pm  \\ 
   \hline
\end{tabular}
\caption{The location of break points for the NYC Yellow Cab Demand data.}\label{table_Taxi}
\end{table}

\section{Discussion}\label{sec:disc}

In this article, we developed a two-stage method for detecting structural break point in high-dimensional piecewise stationary VAR models. A block coordinate descent algorithm was developed to implement the proposed method efficiently. 

We showed that the proposed method consistently detects the total number of the break points, as well as their locations. Numerical experiments through in three simulation settings and two real data applications corroborate these theoretical  findings. In particular, in both real data sets, the break points detected using the proposed method are in agreement with the nature of the data sets. 

When the total number of break points $ m_0 $ is finite, the rate of consistency for detecting break point locations relative to the sample size $ T $ is affected by three factors: (1) the number of time points $ T $, (2) the number of time series observed $ p $, (3) the total sparsity of the model $ d_n^\star $. For the univariate case, this rate was shown to be of order $ ({\log n})/{n} $ by \cite{Chan_2014}. In the high-dimensional case, the rate is shown here to be of order $ ({d_n^\star \log n \log p })/{n} $. This rate puts an upper bound on the number of time series observed and the total sparsity in the model in the high-dimensional setting. Moreover, the proposed procedure allows for the number of break points to increase with the sample size, as long as the minimum distance between consecutive break points is large enough (Assumptions $ A3 $ and $ A4 $ connect the consistency rate of break point detection with the minimum distance between consecutive break points). Extending the methodology and theory in this paper to high-dimensional threshold autoregressive (TAR) models \citep{Tsay_1989} can offer an interesting direction of future research.

\section*{Appendix}\label{sec:appendix}
This section collects the technical lemmas, as well as the proofs of the main results in the paper. 

\subsection*{Appendix~A: Technical Lemmas}

\begin{lemma}\label{lemma_first}
There exist constants $ c_i > 0 $ such that for $ n \geq c_0 \left( \log(n) + 2\log(p) + \log(d)  \right) $, with probability at least $ 1 - c_1 \exp \left( - c_2 \left( \log(n) + 2\log(p) + \log(d)  \right) \right) $, we have 
\begin{equation}
\left|\left|  \frac{Z^\prime E }{n}  \right|\right|_\infty \leq c_3 \sqrt{\frac{\log(n) + 2\log(p) + \log(d)}{n}}
\end{equation}
\end{lemma}

\begin{proof}
Note that $ \frac{1}{n} Z^\prime E = \frac{1}{n} (I_p \otimes \mathcal{X}^\prime) E = \mbox{vec} (\mathcal{X}^\prime \varepsilon(n))/n $. Let $ \mathcal{X}(h,.) $ and $ \mathcal{X}(h,l) $ be the $ h$-th block column and the $ l$-th column of the $ h$-th block column of $ \mathcal{X} $, respectively, $ 1 \leq h \leq n $, $ 1 \leq l \leq d $. More specifically, 

\begin{equation}
\mathcal{X}(h,.) = \begin{pmatrix}  & 0 & \\ & \vdots & \\ & 0 & \\  y_{d+h-2}^\prime & \ldots & y_{h-1}^\prime \\ & \vdots & \\  y_{T-1}^\prime & \ldots & y_{T-d}^\prime \end{pmatrix}_{n \times pd},  \hspace{1cm} \mathcal{X}(h,l) = \begin{pmatrix} 0 \\  \vdots \\ 0 \\  y_{d+h-l-1}^\prime \\  \vdots  \\ y_{T-l}^\prime \end{pmatrix}_{n \times p}. 
\end{equation} 
Now, 
\begin{equation}
\left|\left|  \frac{Z^\prime E }{n}  \right|\right|_\infty = \max_{1 \leq h \leq n, 1 \leq l \leq d, 1 \leq i,j \leq p} \left|  e_i^\prime \left( \frac{\mathcal{X}^\prime (h,l) \varepsilon(n) }{n}  \right) e_j  \right|,
\end{equation}
where $ e_i \in \mathbb{R}^p $ with the $ i$-th element equals to 1 and zero on the rest. Note that,

$$  \frac{\mathcal{X}^\prime (h,l) \varepsilon(n) }{n} = \frac{1}{n} \sum_{j=h-l-1}^{T-d-l} y_{d+j} \varepsilon^\prime_{d+j+l}. $$

Now. sine $ \mbox{cov} (y_{d+j},  \varepsilon_{d+j+l}) = 0 $ for all $ j, l, h $, similar argument as in proposition 2.4 (b) of (\cite{Basu_2015}) shows that for fixed $ i, j, h, l $, there exist $ k_1, k_2 > 0 $ such that for all $ \eta > 0 $:

$$ \mathbb{P} \left( \left|  e_i^\prime \left( \frac{\mathcal{X}^\prime (h,l) \varepsilon(n) }{n}  \right) e_j  \right| > k_1 \eta  \right) \leq 6 \exp \left(  - k_2 n \min (\eta, \eta^2)  \right). $$

Set $ \eta = k_3  \sqrt{\frac{\log(n) + 2\log(p) + \log(d)}{n}} $ for a large enough $ k_3 > 0 $, and taking the union over the $ q = n p^2 d $ possible choices of $ i, j, h, l $ yield the result.  
\end{proof}

\begin{lemma}\label{lemma_KKT}
Let $ \widehat{\Theta} $ be defined as in (\ref{eq_estimation}), then under the assumptions of theorem (\ref{thm_pred_error}):
\begin{equation}
\sum_{l = \widehat{t}_j}^{n} Y_{l-1} \left( y_l^\prime - Y_{l-1}^\prime \sum_{i=1}^{l} \widehat{\theta}_i^\prime \right) = \frac{n \lambda_n}{2} \mbox{sign} (\widehat{\theta}_{\widehat{t}_j}^\prime), \hspace{1cm} \mbox{for} \,\,\,  j = 1, 2, ..., \widehat{m},
\end{equation}
where $ Y_l^\prime = \left( y_l^\prime \ldots y_{l-d+1}^\prime  \right)_{1 \times pd} $, and 
\begin{equation}
\left|\left|  \sum_{l = j}^{n} Y_{l-1} \left( y_l^\prime - Y_{l-1}^\prime \sum_{i=1}^{l} \widehat{\theta}_i^\prime \right) \right|\right|_\infty \leq \frac{n \lambda_n}{2}, \hspace{1cm} \mbox{for} \,\,\,  j = d-1, 2, ..., n.
\end{equation}
Moreover, $ \sum_{i=1}^{t} \widehat{\theta}_i = \widehat{\Phi}^{(.,j)} $ for $ \widehat{t}_{j-1} \leq t \leq \widehat{t}_j - 1$, $ j = 1, 2, ..., |  \mathcal{A}_n| $. 
\end{lemma}
\begin{proof}
This is just checking the KKT condition of the proposed optimization problem.
\end{proof}

\begin{lemma}\label{lemma_bound}
Under assumption A1, there exist constants $ c_i > 0 $ such that with probability at least $ 1 - c_1 \exp(-c_2 (\log(d) + 2 \log(p))  ) $, 

\begin{equation}
\sup_{1 \leq i \leq m_0, s \geq t_i, |  t_i - s  |  > n \gamma_n  }   \left|\left|  {(t_i - s)}^{-1} \left( \sum_{l=s}^{t_i - 1} Y_{l-1} Y_{l-1}^\prime - \Gamma_i^d (0) \right)  \right|\right|_\infty \leq c_3 \sqrt{\frac{\log(d) + 2 \log(p)}{n \gamma_n}},
\end{equation}
where $ \Gamma_i^d (0) = \mathbb{E} (Y_{l-1} Y_{l-1}^\prime) $, and 
\begin{equation}
\sup_{1 \leq i \leq m_0, s \geq t_i, |  t_i - s  |  > n \gamma_n  }  \left|\left|  {(t_i - s)}^{-1} \sum_{l =  s }^{t_i - 1} Y_{l-1} \varepsilon_l^\prime  \right|\right|_\infty \leq c_3 \sqrt{\frac{\log(d) + 2 \log(p)}{n \gamma_n}}.
\end{equation}

\end{lemma}

\begin{proof} 
The proof of this lemma is similar to proposition 2.4 in \cite{Basu_2015}. Here we briefly mention the proof omitting the details. For the first one, note that using similar argument as in proposition 2.4 (a) in \cite{Basu_2015}, there exist $ k_1, k_2 > 0 $ such that for each fixed $ i, j = 1, ..., pd $,
\begin{equation}
\mathbb{P} \left( \left|  e_i^\prime  \frac{ \sum_{l=s}^{t_i - 1} Y_{l-1} Y_{l-1}^\prime - \Gamma_i^d (0) }{t_i - s} e_j  \right| > k_1 \eta \right) \leq 6 \exp (-k_2 n \gamma_n \min(\eta, \eta^2) ).
\end{equation}
Setting $ \eta = k_3 \sqrt{\frac{\log(d p^2)}{n \gamma_n}} $, and taking the union over all possible values of $ i, j $, we get the first part. For the second part, the proof will be similar to lemma (\ref{lemma_first}). Again, there exist $ k_1, k_2 > 0 $ such that for each fixed $ i = 1, ..., pd $, $ j = 1, ..., p $,

\begin{equation}
\mathbb{P} \left( \left|  e_i^\prime  \frac{ \sum_{l =  s }^{t_i - 1} Y_{l-1} \varepsilon_l^\prime  }{t_i - s} e_j  \right| > k_1 \eta \right) \leq 6 \exp (-k_2 n \gamma_n \min(\eta, \eta^2) ).
\end{equation}

Setting $ \eta = k_3 \sqrt{\frac{\log(d p^2)}{n \gamma_n}} $, and taking the union over all possible values of $ i, j $, we get:

\begin{equation}
 \left|\left|  {(t_i - s)}^{-1} \left( \sum_{l=s}^{t_i - 1} Y_{l-1} Y_{l-1}^\prime - \Gamma_i^d (0) \right)  \right|\right|_\infty \leq c_3 \sqrt{\frac{\log(d) + 2 \log(p)}{n \gamma_n}},
\end{equation}
and 
\begin{equation}
 \left|\left|  {(t_i - s)}^{-1} \sum_{l =  s }^{t_i - 1} Y_{l-1} \varepsilon_l^\prime  \right|\right|_\infty \leq c_3 \sqrt{\frac{\log(d) + 2 \log(p)}{n \gamma_n}},
\end{equation}
with high probability converging to 1 for any $ i = 1, 2, ..., m_0 $, as long as $ |  t_i - s  |  > n \gamma_n $ and $ s \geq t_{i-1} $. 
Note that the constants $ c_1, c_2, c_3 $ can be chosen large enough and in such a way that the upper bounds above would be independent of the break point $ t_i $. Therefore, we have the desired upper bounds verified with probability at least $ 1 - c_1 \exp(-c_2 (\log(d) + 2 \log(p))  ) $. 
\end{proof}

\begin{lemma}\label{lemma_selection}
Under the assumptions of theorem (\ref{thm_selection}), for $ m < m_0 $, there exists a constant $ c > 0 $ such that:

\begin{equation}\label{eq_lower_bound}
\mathbb{P} \left( \min_{(s_1, ..., s_m) \subset \lbrace 1, ..., T \rbrace} L_n(s_1, s_2, ..., s_m; \eta_n)  >  \sum_{t=d}^{T} || \varepsilon_t ||_2^2  + c \Delta_n - M_{\Phi} (m_0+1) n \gamma_n d_n^\star \right) \rightarrow 1,
\end{equation}
where $ \Delta_n =  \min_{1 \leq j \leq m_0+1} | t_j - t_{j-1} |  $. 
\end{lemma}
\begin{proof} 
Since $ m < m_0 $, there exists a point $ t_j $ such that $ | s_i - t_j | > \Delta_n/4 $. Now, $ L_n(s_1, s_2, ..., s_m; \eta_n) -  n \, \eta_n \sum_{i=1}^{m+1} || \widehat{\theta}_{(s_{i-1},s_i)} ||_1 =  || \mathcal{Y} - \mathcal{X}_{s_1, ..., s_m} \widehat{\theta}_{s_1, ..., s_m} ||_F^2  =  \sum_{i=1}^{m_0+2} T_i  $, where $ T_i $ is the sum of squares involving $ Y_k $, $ t_{i-1} \leq k < t_i $ for $ i = 1, ..., j-1, j+2, ..., m_0+1 $, and $ T_j, T_{j+1}, T_{m_0+2} $ are the sums of $ Y_k $ for $ t_{j-1} \leq k < t_j - \Delta_n/4 $, $ t_j + \Delta_n/4 \leq k < t_{j+1} $, and $ t_j - \Delta_n/4 \leq k < t_j + \Delta_n/4 $, respectively. For $ i = 1, ..., j-1, j+2, ..., m_0+1 $, we find a lower bound for the $ T_i $. For a fixed $ i $, let's say there are $ r_i $ points within $ [t_{i-1}, t_i) $, denoting them by $ \lbrace s_l, s_{l+1}, ..., s_{l+r_i} \rbrace \subset \lbrace s_1, ..., s_m \rbrace $, we put $ r_i = -1 $ if there are no points. Now, $ T_i $ can be decomposed as:

\begin{equation}
T_i = \sum_{t=t_{i-1}}^{s_l-1} || y_t - \widehat{\theta}_{(t_{i-1},s_l)} Y_{t-1}   ||_2^2 + \sum_{h=l}^{l+r_i-1} \sum_{t=s_h}^{s_{h+1}-1} || y_t - \widehat{\theta}_{(s_h,s_{h+1})} Y_{t-1}   ||_2^2 + \sum_{t=s_{l+r_i}}^{t_i-1} || y_t - \widehat{\theta}_{(s_{l+r_i},t_i)} Y_{t-1}   ||_2^2.
\end{equation}
Note that for a fixed $ h $, 
\begin{eqnarray}
\sum_{t=s_h}^{s_{h+1}-1} || y_t - \widehat{\theta}_{(s_h,s_{h+1})} Y_{t-1}   ||_2^2 &=& \sum_{t=s_h}^{s_{h+1}-1} || \varepsilon_t ||_2^2 + \sum_{t=s_h}^{s_{h+1}-1} || (\Phi^{(.,i)} -  \widehat{\theta}_{(s_h,s_{h+1})} )  Y_{t-1}  ||_2^2 \nonumber \\
& + & 2 \sum_{t=s_h}^{s_{h+1}-1} Y_{t-1}^\prime (\Phi^{(.,i)} -  \widehat{\theta}_{(s_h,s_{h+1})} )^\prime \varepsilon_t.
\end{eqnarray}
Now, by similar arguments as in lemma (\ref{lemma_bound}), we have:

\begin{eqnarray}\label{eq_error_rate}
\left| \sum_{t=s_h}^{s_{h+1}-1} Y_{t-1}^\prime (\Phi^{(.,i)} -  \widehat{\theta}_{(s_h,s_{h+1})} )^\prime \varepsilon_t \right| & \leq & \frac{|| \sum_{t=s_h}^{s_{h+1}-1} Y_{t-1} \varepsilon_t^\prime  ||_\infty}{n \gamma_n} n \gamma_n || \Phi^{(.,i)} -  \widehat{\theta}_{(s_h,s_{h+1})}  ||_1 \nonumber \\
&=& o_p (n \gamma_n) || \Phi^{(.,i)} -  \widehat{\theta}_{(s_h,s_{h+1})}  ||_1,
\end{eqnarray}
which adds up to:
\begin{equation}
T_i \geq \sum_{t = t_{i-1}}^{t_i-1} || \varepsilon_t ||_2^2 - o_p (n \gamma_n) \left( \sum_{h=l}^{l+r_i-1}  || \Phi^{(.,i)} -  \widehat{\theta}_{(s_h,s_{h+1})}  ||_1 +  || \Phi^{(.,i)} -  \widehat{\theta}_{(t_{i-1},s_l)}  ||_1 +  || \Phi^{(.,i)} -  \widehat{\theta}_{(s_{l+r_i},t_i)}  ||_1  \right).
\end{equation}

Let's focus on $ T_{m_0+2} $.  Since there are no points inside the interval $ [t_j - \Delta_n/4, t_j + \Delta_n/4) $, $ \widehat{\theta}_{(t_j - \Delta_n/4, t_j)} =  \widehat{\theta}_{(t_j,t_j + \Delta_n/4)} = \theta^\star $. Now, we can decompose it as:

\begin{equation}
T_{m_0+2} = \sum_{t=t_j - \Delta_n/4}^{t_j-1} || y_t - \theta^\star Y_{t-1}   ||_2^2 + \sum_{t=t_j }^{t_j+ \Delta_n/4} || y_t - \theta^\star Y_{t-1}   ||_2^2 = I + II.
\end{equation}
We zoom in $ I $:

\begin{eqnarray}
I &=&  \sum_{t=t_j - \Delta_n/4}^{t_j-1} || \varepsilon_t ||_2^2 + \sum_{t=t_j - \Delta_n/4}^{t_j-1} || (\Phi^{(.,j)} -  \theta^\star )  Y_{t-1}  ||_2^2 \nonumber \\
& + & 2 \sum_{t=t_j - \Delta_n/4}^{t_j-1} Y_{t-1}^\prime (\Phi^{(.,j)} -  \theta^\star )^\prime \varepsilon_t \nonumber \\
&=& \sum_{t=t_j - \Delta_n/4}^{t_j-1} || \varepsilon_t ||_2^2 + I_A + I_B.
\end{eqnarray}

Similar to (\ref{eq_error_rate}), we have:

\begin{eqnarray}
\left| I_B \right| & \leq & o_p (n \gamma_n) || \Phi^{(.,j)} -  \theta^\star  ||_1,
\end{eqnarray}

We need to find a large enough bound for $ I_A $. Denote the $ i$-th row of $ \Phi^{(.,j)} -  \theta^\star $ by $ v_i $ for $ i = 1, ..., p $. Now,

\begin{eqnarray}
I_A &=& \sum_{t=t_j - \Delta_n/4}^{t_j-1} Y_{t-1}^\prime (\Phi^{(.,j)} -  \theta^\star)^\prime  (\Phi^{(.,j)} -  \theta^\star) Y_{t-1} \nonumber \\
&=& \mbox{tr} \left(  (\Phi^{(.,j)} -  \theta^\star) \left( \sum_{t=t_j - \Delta_n/4}^{t_j-1} Y_{t-1} Y_{t-1}^\prime \right) (\Phi^{(.,j)} -  \theta^\star)^\prime  \right) \nonumber \\
&=& \sum_{i=1}^{p} v_i \left( \sum_{t=t_j - \Delta_n/4}^{t_j-1} Y_{t-1} Y_{t-1}^\prime \right) v_i^\prime \nonumber \\
&=& \sum_{i=1}^{p} v_i \left( \sum_{t=t_j - \Delta_n/4}^{t_j-1} \left( Y_{t-1} Y_{t-1}^\prime - \Gamma_j (0) \right) + \frac{\Delta_n}{4}   \Gamma_j (0) \right) v_i^\prime
\end{eqnarray}

By similar arguments as in lemma (\ref{lemma_bound}), we have:
$$  \frac{4}{\Delta_n}  \left|\left| \sum_{t=t_j - \Delta_n/4}^{t_j-1} \left( Y_{t-1} Y_{t-1}^\prime - \Gamma_j (0) \right)   \right| \right|_\infty = o_p(1). $$
Using the above fact, 
\begin{equation}
I_A \geq \frac{\Delta_n}{8} \Lambda_{\min} (\Gamma_j(0)) \sum_{i=1}^{p} || v_i ||_2^2 = c_1  || \Phi^{(.,j)} -  \theta^\star ||_2^2,
\end{equation}
with $ c_1 = \frac{1}{8} \min_{1 \leq j \leq m_0+1} \Lambda_{\min} (\Gamma_j(0)) $. All combined lead to:
\begin{equation}
I \geq \sum_{t=t_j - \Delta_n/4}^{t_j-1} || \varepsilon_t ||_2^2 + c_1  || \Phi^{(.,j)} -  \theta^\star ||_2^2 - o_p (n \gamma_n) || \Phi^{(.,j)} -  \theta^\star  ||_1.
\end{equation}
Similarly, one can show that:
\begin{equation}
II \geq \sum_{t=t_j }^{t_j+ \Delta_n/4} || \varepsilon_t ||_2^2 + c_1  || \Phi^{(.,j+1)} -  \theta^\star ||_2^2 - o_p (n \gamma_n) || \Phi^{(.,j+1)} -  \theta^\star  ||_1.
\end{equation}

Now, since $ \min_{1 \leq j \leq m_0} \left|\left| \Phi^{(.,j+1)} - \Phi^{(.,j)}  \right|\right|_2  \geq v > 0 $, we have:
\begin{equation}
T_{m_0+2} = I + II \geq \sum_{t=t_j - \Delta_n/4}^{t_j+ \Delta_n/4} || \varepsilon_t ||_2^2 + c \Delta_n - o_p (n \gamma_n) \left(   || \Phi^{(.,j)} -  \theta^\star  ||_1 + || \Phi^{(.,j+1)} -  \theta^\star  ||_1 \right),
\end{equation}
where $ c = c_1 \nu $. Now, since we don't know which true segments will be inside each estimated segment, we have the following lower bound:
\begin{equation}
L_n(s_1, s_2, ..., s_m; \eta_n) -  n \, \eta_n \sum_{i=1}^{m+1} || \widehat{\theta}_{(s_{i-1},s_i)} ||_1 =  \sum_{i=1}^{m_0+2} T_i  \geq \sum_{t=d}^{T} || \varepsilon_t ||_2^2  + c \Delta_n - o_p (n \gamma_n) \sum_{i=1}^{m+1} \sum_{j=1}^{m_0+1} || \Phi^{(.,j)} -  \widehat{\theta}_{(s_{i-1},s_i)} ||_1.
\end{equation}
 Now, by assumption A4 (a) or (b), we have:
\begin{eqnarray}
L_n(s_1, s_2, ..., s_m; \eta_n) & \geq & \sum_{t=d}^{T} || \varepsilon_t ||_2^2 + c \Delta_n - (m_0+1) n \gamma_n \sum_{j=1}^{m_0+1} || \Phi^{(.,j)} ||_1 \nonumber \\
& \geq & \sum_{t=d}^{T} || \varepsilon_t ||_2^2 + c \Delta_n - M_{\Phi} (m_0+1) n \gamma_n d_n^\star,
\end{eqnarray}
with high probability approaching to 1. Note that the lower bound doesn't depend on the choices of $ s_i$'s as long as $ m < m_0 $. This completes the proof. 
\end{proof}

\subsection*{Appendix~B: Proof of Main Results}

\begin{proof}[Proof of Theorem~\ref{thm_pred_error}]
By definition of $ \widehat{\Theta} $, we get
\begin{equation}
\frac{1}{n} || Y - Z \widehat{\Theta}  ||_2^2 + \lambda_n \sum_{i=1}^{n} || \widehat{\theta}_i ||_1 \leq \frac{1}{n} || Y - Z \Theta ||_2^2 + \lambda_n \sum_{i=1}^{n} || \theta_i ||_1.
\end{equation}
Denoting $ \mathcal{A} = \lbrace t_1, t_1, ..., t_{m_0}  \rbrace $, we have:
\begin{eqnarray}
\frac{1}{n} \left|\left| Z \left( \widehat{\Theta} - \Theta  \right)  \right|\right|_2^2 &\leq & \frac{2}{n} \left(  \widehat{\Theta} - \Theta \right)^\prime Z^\prime E + \lambda_n \sum_{i=1}^{n} || \widehat{\theta}_i ||_1 - \lambda_n \sum_{i=1}^{n} || \theta_i ||_1 \nonumber \\ & \leq & 2 \left|\left|  \frac{Z^\prime E }{n}  \right|\right|_\infty \sum_{i=1}^{n}  || \theta_i - \widehat{\theta}_i ||_1 +\lambda_n \sum_{i \in \mathcal{A}} \left( || \theta_i ||_1 -  || \widehat{\theta}_i ||_1  \right) - \lambda_n \sum_{i \in \mathcal{A}^c} || \widehat{\theta}_i ||_1 \nonumber \\ 
& = & \lambda_n \sum_{i \in \mathcal{A}}  || \theta_i - \widehat{\theta}_i ||_1 + \lambda_n \sum_{i \in \mathcal{A}} \left( || \theta_i ||_1 -  || \widehat{\theta}_i ||_1  \right) \nonumber \\
& \leq & 2 \lambda_n \sum_{i \in \mathcal{A}} || \theta_i ||_1 \nonumber \\
& \leq & 2 \lambda_n m_n  \max_{1 \leq j \leq m_0+1} \left|\left| \Phi^{(.,j)} - \Phi^{(.,j-1)} \right|\right|_1 \nonumber \\
& = & 4 C m_n  \max_{1 \leq j \leq m_0+1} \left\lbrace \sum_{i=1}^{p} \left( s_{ij} + s_{i(j-1)}  \right) \right\rbrace   M_\Phi \sqrt{\frac{\log(n) + 2\log(p) + \log(d)}{n}},
\end{eqnarray}
with high probability approaching to 1 due to the lemma (\ref{lemma_first}). 
\end{proof}

\begin{proof}[Proof of Theorem~\ref{thm_known_m}]

The proof is similar to theorem 2.2 in (\cite{Chan_2014}) and proposition 5 in (\cite{Harchaoui_2010}). Before we start, define for a matrix $ A \in \mathbb{R}^{pd \times p} $, $ || A ||_{\infty, NZ} = \max_{j \in NZ, 1 \leq i \leq p} |a_{ji}| $. Now, if for some $ i = 1, ..., m_0 $, $ | \widehat{t}_i - t_i | > n \gamma_n $, this means that there exists a true break point $ t_{i_0+1} $ which is isolated from all the estimated points, i.e. $ \min_{1 \leq i \leq m_0} | \widehat{t}_i - t_{i_0+1} | > n \gamma_n  $. In other words, there exists an estimated break point $ \widehat{t}_j $ such that, $ t_{i_0+1} - t_{i_0} \vee \widehat{t}_j \geq n \gamma_n $ and $ t_{i_0+2} \wedge \widehat{t}_{j+1} \geq n \gamma_n $. Apply lemma (\ref{lemma_KKT}) twice to get:

\begin{equation}\label{eq:thm_known_first}
\left|\left| \sum_{l = t_{i_0} \vee \widehat{t}_j }^{t_{i_0+1} - 1} Y_{l-1} Y_{l-1}^\prime \left( {\Phi^\prime}^{(.,i_0+1)} - \widehat{\Phi^\prime}^{(.,j+1)}  \right) \right|\right|_{\infty,NZ} \leq n \lambda_n + \left|\left|  \sum_{l = t_{i_0} \vee \widehat{t}_j }^{t_{i_0+1} - 1} Y_{l-1} \varepsilon_l^\prime \right|\right|_\infty
\end{equation}

and 

\begin{equation}\label{eq:thm_known_second}
\left|\left| \sum_{l = t_{i_0+1} }^{t_{i_0+2} \wedge \widehat{t}_{j+1} - 1} Y_{l-1} Y_{l-1}^\prime \left( {\Phi^\prime}^{(.,i_0+2)} - \widehat{\Phi^\prime}^{(.,j+1)}  \right) \right|\right|_{\infty,NZ} \leq n \lambda_n + \left|\left|  \sum_{l = t_{i_0+1} }^{t_{i_0+2} \wedge \widehat{t}_{j+1} - 1} Y_{l-1} \varepsilon_l^\prime \right|\right|_\infty.
\end{equation}

Now, consider the first equation (\ref{eq:thm_known_first}). We can write the left hand side as 

\begin{eqnarray}
(t_{i_0+1} - t_{i_0} \vee \widehat{t}_j)^{-1} \left|\left| \sum_{l = t_{i_0} \vee \widehat{t}_j }^{t_{i_0+1} - 1} Y_{l-1} Y_{l-1}^\prime \left( {\Phi^\prime}^{(.,i_0+1)} - \widehat{\Phi^\prime}^{(.,j+1)}  \right) \right|\right|_{\infty,NZ} &\geq & \left|\left| (\Gamma_i^d (0) - A) \left( {\Phi^\prime}^{(.,i_0+1)}  \right) \right|\right|_{\infty,NZ} \nonumber \\
& - & \left|\left| (\Gamma_i^d (0) - A) \left(  \widehat{\Phi^\prime}^{(.,j+1)}  \right) \right|\right|_{\infty},
\end{eqnarray}

for some random matrix $ A $ with $ ||A||_\infty \rightarrow 0 $ with high probability converging to one based on lemma (\ref{lemma_bound}). Then, we can show that based on the properties of the covariance matrix $ \Gamma_i^d (0) $ that:

\begin{equation}
\left|\left| (\Gamma_i^d (0) - A) \left( {\Phi^\prime}^{(.,i_0+1)}  \right) \right|\right|_{\infty,NZ} \geq c_1 (s^\star)^{-1} \max_{1 \leq i \leq p}  \left|\left| \Phi_i^{(.,i_0+1)} \right|\right|_2,
\end{equation} 
and 
\begin{equation}
\left|\left| (\Gamma_i^d (0) - A) \left(  \widehat{\Phi^\prime}^{(.,j+1)}  \right) \right|\right|_{\infty} \leq c_2 \left|\left|  \widehat{\Phi^\prime}^{(.,j+1)} \right|\right|_1,
\end{equation}
for some positive constants $ c_1, c_2 $. 
Putting them all together, and use lemma (\ref{lemma_bound}) again for for the second term on the right hand side of equation (\ref{eq:thm_known_first}), we have:

\begin{equation}
c_1  \max_{1 \leq i \leq p}  \left|\left| \Phi_i^{(.,i_0+1)} \right|\right|_2 - c_2 s^\star  \left|\left|  \widehat{\Phi^\prime}^{(.,j+1)} \right|\right|_1 \leq \frac{s^\star n \lambda_n}{(t_{i_0+1} - t_{i_0} \vee \widehat{t}_j)} + k_1 s^\star \sqrt{\frac{\log p}{n \gamma_n}}.
\end{equation}
 The right hand side goes to zero based on A2 and A3. Similarly, we can use equation (\ref{eq:thm_known_second}) to show that

\begin{equation}
c_3  \max_{1 \leq i \leq p}  \left|\left| \Phi_i^{(.,i_0+2)} \right|\right|_2 - c_4 s^\star  \left|\left|  \widehat{\Phi^\prime}^{(.,j+1)} \right|\right|_1 \leq \frac{s^\star n \lambda_n}{(t_{i_0+1} - t_{i_0} \vee \widehat{t}_j)} + k_1 s^\star \sqrt{\frac{\log p}{n \gamma_n}}.
\end{equation}

Putting them together implies that:

\begin{equation}
\frac{\max_{1 \leq i \leq p}  \left|\left| \Phi_i^{(.,i_0+2)} \right|\right|_2}{\max_{1 \leq i \leq p}  \left|\left| \Phi_i^{(.,i_0+1)} \right|\right|_2} \leq c_5,
\end{equation}

and so, if we choose the $ \nu^\prime $ large enough in A2, we reach the contradiction. This completes the proof.

\end{proof}

\begin{proof}[Proof of Theorem~\ref{thm_Hausdorff}]
The proof is similar to the proof of theorem 2.3 in \cite{Chan_2014}. Here we will mention the proof of the first part. For that, assume $ |\mathcal{A}_n | < m_0 $. This means there exist an isolated true break point, say $ t_{i_0} $. More specifically, there exists an estimated break point $ \widehat{t}_j $ such that, $ t_{i_0+1} - t_{i_0} \vee \widehat{t}_j \geq n \gamma_n / 3 $ and $ t_{i_0+2} \wedge \widehat{t}_{j+1} \geq n \gamma_n / 3 $. Apply lemma (\ref{lemma_KKT}) twice to get:

\begin{equation}
\left|\left| \sum_{l = t_{i_0} \vee \widehat{t}_j }^{t_{i_0+1} - 1} Y_{l-1} Y_{l-1}^\prime \left( {\Phi^\prime}^{(.,i_0+1)} - \widehat{\Phi^\prime}^{(.,j+1)}  \right) \right|\right|_{\infty,NZ} \leq n \lambda_n + \left|\left|  \sum_{l = t_{i_0} \vee \widehat{t}_j }^{t_{i_0+1} - 1} Y_{l-1} \varepsilon_l^\prime \right|\right|_\infty
\end{equation}

and 

\begin{equation}
\left|\left| \sum_{l = t_{i_0+1} }^{t_{i_0+2} \wedge \widehat{t}_{j+1} - 1} Y_{l-1} Y_{l-1}^\prime \left( {\Phi^\prime}^{(.,i_0+2)} - \widehat{\Phi^\prime}^{(.,j+1)}  \right) \right|\right|_{\infty,NZ} \leq n \lambda_n + \left|\left|  \sum_{l = t_{i_0+1} }^{t_{i_0+2} \wedge \widehat{t}_{j+1} - 1} Y_{l-1} \varepsilon_l^\prime \right|\right|_\infty.
\end{equation}
Now, similar argument as in theorem (\ref{thm_known_m}) reaches to contradiction, and this completes the proof.
\end{proof}

\begin{proof}[Proof of Theorem~\ref{thm_selection}]
Let's focus on the first part. We show that (a) $ \mathbb{P} ( \dhat{m} < m_0 ) \rightarrow 0 $, and (b) $ \mathbb{P} ( \dhat{m} > m_0 ) \rightarrow 0 $. For the first claim, from theorem (\ref{thm_Hausdorff}), we know that there are points $ \widehat{t}_i \in \mathcal{A}_n $ such that $ \max_{1 \leq i \leq m_0} | \widehat{t}_i - t_i  | \leq n \gamma_n $. The parameter estimated when choosing these $ m_0 $ points are $ \widehat{\theta}_{(\widehat{t}_1, ..., \widehat{t}_{m_0})} $. By the definition of this parameter, it minimizes the least squares plus the $ L_1 $ penalty on its norm. Therefore, it has to beat the case where one puts $ \Phi^{(.,j)} $ on the segment $ [\widehat{t}_{j-1}, \widehat{t}_{j})$ for $ j = 1, ..., m_0+1 $. This leads to an upper bound for $ L(\widehat{t}_1, ..., \widehat{t}_{m_0} ;\eta_n) $. By similar arguments as in lemma (\ref{lemma_selection}), we get that there exist constants $ K_1, K_2, K > 0 $ such that:

\begin{eqnarray}
L(\widehat{t}_1, ..., \widehat{t}_{m_0} ;\eta_n) & \leq & \sum_{t=d}^{T} || \varepsilon_t ||_2^2 + o_p(n \gamma_n) \sum_{j=1}^{m_0} \left|\left| \Phi^{(.,j+1)} - \Phi^{(.,j)}  \right| \right|_1 \nonumber \\
& + & K_1 n \gamma_n \sum_{j=1}^{m_0} \left|\left| \Phi^{(.,j+1)} - \Phi^{(.,j)}  \right| \right|_2^2 + K_2 n \gamma_n \sum_{j=1}^{m_0+1} \left|\left| \Phi^{(.,j+1)} \right| \right|_1 \nonumber \\
& \leq & \sum_{t=d}^{T} || \varepsilon_t ||_2^2 + K n \gamma_n d_n^\star.
\end{eqnarray}

Now,

\begin{eqnarray}
IC( \dhat{t}_1, ..., \dhat{t}_{\dhat{m}} ) & = & L_n (\dhat{t}_1, ..., \dhat{t}_{\dhat{m}}; \eta_n ) + \dhat{m} \omega_n \nonumber \\
& > & \sum_{t=d}^{T} || \varepsilon_t ||_2^2 + c \Delta_n - M_{\Phi} (m_0+1) n \gamma_n d_n^\star + \dhat{m} \omega \nonumber \\
& \geq & L(\widehat{t}_1, ..., \widehat{t}_{m_0} ;\eta_n) + m_0 \omega_n + c \Delta_n - K_3 (m_0+1) n \gamma_n d_n^\star - (m_0 -  \dhat{m}) \omega_n \nonumber \\
& \geq & L(\widehat{t}_1, ..., \widehat{t}_{m_0} ;\eta_n) + m_0 \omega_n,
\end{eqnarray}
since $ \lim_{n \rightarrow \infty} n \gamma_n d_n^\star / \omega_n \leq 1 $, and $ \lim_{n \rightarrow \infty} m_0 \omega_n / \Delta_n = 0 $. This proves part (a). To prove part (b), note that a similar argument as in lemma (\ref{lemma_selection}) shows that 
\begin{equation}
L_n ( \dhat{t}_1, ..., \dhat{t}_{\dhat{m}}  ; \eta_n) \geq \sum_{t=d}^{T} || \varepsilon_t ||_2^2 - K_4 m n \gamma_n d_n^\star,
\end{equation}
for some constant $ K_4 > 0 $. A comparison between $ IC (\dhat{t}_1, ..., \dhat{t}_{\dhat{m}}) $ and $ IC (\widehat{t}_1, ..., \widehat{t}_{m_0}) $ yields to:

\begin{eqnarray}
\sum_{t=d}^{T} || \varepsilon_t ||_2^2 - K_4 m n \gamma_n d_n^\star + m \omega_n & \leq & IC (\dhat{t}_1, ..., \dhat{t}_{\dhat{m}}) \nonumber \\
& \leq & IC (\widehat{t}_1, ..., \widehat{t}_{m_0}) \nonumber \\
& \leq & \sum_{t=d}^{T} || \varepsilon_t ||_2^2 + K n \gamma_n d_n^\star + m_0 \omega_n,
\end{eqnarray}
which means:
\begin{equation}
(m - m_0) \omega_n \leq K_4 m n \gamma_n d_n^\star + K n \gamma_n d_n^\star,
\end{equation}
which contradicts with the fact that $ m_0 n \gamma_n d_n^\star / \omega_n \rightarrow 0 $. This completes the first part of the theorem. 

For the second part, put $ B = 2 K / c $. Now, suppose that there exists a point $ t_i $ such that $ \min_{1 \leq j \leq m_0} | \dhat{t}_j - t_j | \geq B n \gamma_n d_n^\star $. Then, by similar argument as in lemma (\ref{lemma_selection}), we can show that:
\begin{eqnarray}
\sum_{t=d}^{T} || \varepsilon_t ||_2^2 + c B n \gamma_n d_n^\star & < & L_n ( \dhat{t}_1, ..., \dhat{t}_{m_0} ) \nonumber \\
& \leq & L_n ( \widehat{t}_1, ..., \widehat{t}_{m_0} ) \nonumber \\
& \leq & \sum_{t=d}^{T} || \varepsilon_t ||_2^2 + K n \gamma_n d_n^\star, 
\end{eqnarray} 
which contradicts with the way $ B $ was selected. This completes the proof of the whole theorem.
\end{proof}


\bibliography{Safikhani_Reference}

\begin{thebibliography}{}

\bibitem[\protect\citename{Bai, }1997]{Bai_1997}
Bai, Jushan. 1997.
\newblock Estimation of a change point in multiple regression models.
\newblock {\em The review of economics and statistics}, {\bf 79}(4), 551--563.

\bibitem[\protect\citename{Basu \& Michailidis, }2015]{Basu_2015}
Basu, Sumanta, \& Michailidis, George. 2015.
\newblock Regularized estimation in sparse high-dimensional time series models.
\newblock {\em The Annals of Statistics}, {\bf 43}(4), 1535--1567.

\bibitem[\protect\citename{Bickel {\em et~al.}, }2009]{BickelETAL_2009}
Bickel, Peter~J, Ritov, Ya�acov, \& Tsybakov, Alexandre~B. 2009.
\newblock Simultaneous analysis of {LASSO} and {Dantzig} selector.
\newblock {\em The Annals of Statistics}, {\bf 37}(4), 1705--1732.

\bibitem[\protect\citename{Boysen {\em et~al.}, }2009]{Boysen_2009}
Boysen, Leif, Kempe, Angela, Liebscher, Volkmar, Munk, Axel, \& Wittich, Olaf.
  2009.
\newblock Consistencies and rates of convergence of jump-penalized least
  squares estimators.
\newblock {\em The Annals of Statistics},  157--183.

\bibitem[\protect\citename{Chan {\em et~al.}, }2014]{Chan_2014}
Chan, Ngai~Hang, Yau, Chun~Yip, \& Zhang, Rong-Mao. 2014.
\newblock Group LASSO for structural break time series.
\newblock {\em Journal of the American Statistical Association}, {\bf
  109}(506), 590--599.

\bibitem[\protect\citename{Chen {\em et~al.}, }2016]{ChenShojaieWitten_2016}
Chen, Shizhe, Shojaie, Ali, \& Witten, Daniela~M. 2016.
\newblock Network Reconstruction From High Dimensional Ordinary Differential
  Equations.
\newblock {\em Journal of the American Statistical Association}.

\bibitem[\protect\citename{Chen {\em et~al.}, }2017]{ChenWittenShojaie_2017}
Chen, Shizhe, Witten, Daniela, Shojaie, Ali, {\em et~al.} 2017.
\newblock Nearly assumptionless screening for the mutually-exciting
  multivariate Hawkes process.
\newblock {\em Electronic Journal of Statistics}, {\bf 11}(1), 1207--1234.

\bibitem[\protect\citename{Chen {\em et~al.}, }2013]{chen2013covariance}
Chen, Xiaohui, Xu, Mengyu, \& Wu, Wei~Biao. 2013.
\newblock Covariance and precision matrix estimation for high-dimensional time
  series.
\newblock {\em The Annals of Statistics}, {\bf 41}(6), 2994--3021.

\bibitem[\protect\citename{Clarida {\em et~al.},
  }2000]{ClaridaGaliGertler_2000}
Clarida, Richard, Gali, Jordi, \& Gertler, Mark. 2000.
\newblock Monetary policy rules and macroeconomic stability: evidence and some
  theory.
\newblock {\em The Quarterly journal of economics}, {\bf 115}(1), 147--180.

\bibitem[\protect\citename{Dahlhaus, }2012]{Dahlhaus_2012}
Dahlhaus, Rainer. 2012.
\newblock Locally stationary processes.
\newblock {\em Handbook of statistics}, {\bf 30}, 351--412.

\bibitem[\protect\citename{Davis {\em et~al.}, }2006]{Davis_2006}
Davis, Richard~A, Lee, Thomas C~M, \& Rodriguez-Yam, Gabriel~A. 2006.
\newblock Structural break estimation for nonstationary time series models.
\newblock {\em Journal of the American Statistical Association}, {\bf
  101}(473), 223--239.

\bibitem[\protect\citename{De~Mol {\em et~al.}, }2008]{de_2008}
De~Mol, Christine, Giannone, Domenico, \& Reichlin, Lucrezia. 2008.
\newblock Forecasting using a large number of predictors: Is Bayesian shrinkage
  a valid alternative to principal components?
\newblock {\em Journal of Econometrics}, {\bf 146}(2), 318--328.

\bibitem[\protect\citename{Ding {\em et~al.}, }2016]{DingQiuChen_2016}
Ding, Xin, Qiu, Ziyi, \& Chen, Xiaohui. 2016.
\newblock Sparse transition matrix estimation for high-dimensional and locally
  stationary vector autoregressive models.
\newblock {\em arXiv preprint arXiv:1604.04002}.

\bibitem[\protect\citename{Fan {\em et~al.}, }2011]{fan_2011sparse}
Fan, Jianqing, Lv, Jinchi, \& Qi, Lei. 2011.
\newblock Sparse high-dimensional models in economics.

\bibitem[\protect\citename{Fujita {\em et~al.}, }2007]{fujita2007modeling}
Fujita, Andr{\'e}, Sato, Jo{\~a}o~R, Garay-Malpartida, Humberto~M, Yamaguchi,
  Rui, Miyano, Satoru, Sogayar, Mari~C, \& Ferreira, Carlos~E. 2007.
\newblock Modeling gene expression regulatory networks with the sparse vector
  autoregressive model.
\newblock {\em BMC Systems Biology}, {\bf 1}(1), 39.

\bibitem[\protect\citename{Hall {\em et~al.}, }2016]{hall_2016}
Hall, Eric~C, Raskutti, Garvesh, \& Willett, Rebecca. 2016.
\newblock Inference of High-dimensional Autoregressive Generalized Linear
  Models.
\newblock {\em arXiv preprint arXiv:1605.02693}.

\bibitem[\protect\citename{Hansen {\em et~al.}, }2015]{HansenETAL_2015}
Hansen, Niels~Richard, Reynaud-Bouret, Patricia, Rivoirard, Vincent, {\em
  et~al.} 2015.
\newblock Lasso and probabilistic inequalities for multivariate point
  processes.
\newblock {\em Bernoulli}, {\bf 21}(1), 83--143.

\bibitem[\protect\citename{Harchaoui \& L{\'e}vy-Leduc, }2010]{Harchaoui_2010}
Harchaoui, Za{\i}d, \& L{\'e}vy-Leduc, C{\'e}line. 2010.
\newblock Multiple change-point estimation with a total variation penalty.
\newblock {\em Journal of the American Statistical Association}, {\bf
  105}(492), 1480--1493.

\bibitem[\protect\citename{Lu {\em et~al.}, }2011]{LuETAL_2011}
Lu, Tao, Liang, Hua, Li, Hongzhe, \& Wu, Hulin. 2011.
\newblock High-dimensional ODEs coupled with mixed-effects modeling techniques
  for dynamic gene regulatory network identification.
\newblock {\em Journal of the American Statistical Association}, {\bf
  106}(496), 1242--1258.

\bibitem[\protect\citename{Michailidis \& d’Alch{\'e} Buc,
  }2013]{michailidis_2013autoregressive}
Michailidis, George, \& d’Alch{\'e} Buc, Florence. 2013.
\newblock Autoregressive models for gene regulatory network inference:
  Sparsity, stability and causality issues.
\newblock {\em Mathematical biosciences}, {\bf 246}(2), 326--334.

\bibitem[\protect\citename{Mukhopadhyay \& Chatterjee,
  }2006]{mukhopadhyay2006causality}
Mukhopadhyay, Nitai~D, \& Chatterjee, Snigdhansu. 2006.
\newblock Causality and pathway search in microarray time series experiment.
\newblock {\em Bioinformatics}, {\bf 23}(4), 442--449.

\bibitem[\protect\citename{Nicholson {\em et~al.}, }2017]{matteson_2017}
Nicholson, William~B, Matteson, David~S, \& Bien, Jacob. 2017.
\newblock VARX-L: Structured regularization for large vector autoregressions
  with exogenous variables.
\newblock {\em International Journal of Forecasting}, {\bf 33}(3), 627--651.

\bibitem[\protect\citename{Ombao {\em et~al.}, }2005]{OmbaoVonSachsGuo_2005}
Ombao, Hernando, Von~Sachs, Rainer, \& Guo, Wensheng. 2005.
\newblock SLEX analysis of multivariate nonstationary time series.
\newblock {\em Journal of the American Statistical Association}, {\bf
  100}(470), 519--531.

\bibitem[\protect\citename{Primiceri, }2005]{Primiceri_2005}
Primiceri, Giorgio~E. 2005.
\newblock Time varying structural vector autoregressions and monetary policy.
\newblock {\em The Review of Economic Studies}, {\bf 72}(3), 821--852.

\bibitem[\protect\citename{Qiu {\em et~al.}, }2016]{QiuETAL_2016}
Qiu, Huitong, Han, Fang, Liu, Han, \& Caffo, Brian. 2016.
\newblock Joint estimation of multiple graphical models from high dimensional
  time series.
\newblock {\em Journal of the Royal Statistical Society: Series B (Statistical
  Methodology)}, {\bf 78}(2), 487--504.

\bibitem[\protect\citename{Sato {\em et~al.}, }2007]{SatoMorettinETAL_2007}
Sato, Jo{\~a}o~R, Morettin, Pedro~A, Arantes, Paula~R, \& Amaro, Edson. 2007.
\newblock Wavelet based time-varying vector autoregressive modelling.
\newblock {\em Computational Statistics \& Data Analysis}, {\bf 51}(12),
  5847--5866.

\bibitem[\protect\citename{Shojaie \& Michailidis,
  }2010]{ShojaieMichailidis_2010}
Shojaie, A., \& Michailidis, G. 2010.
\newblock {Discovering graphical {G}ranger causality using the truncating lasso
  penalty}.
\newblock {\em Bioinformatics}, {\bf 26}(18), i517--i523.

\bibitem[\protect\citename{Shojaie {\em et~al.},
  }2012]{ShojaieBasuMichailidis_2012}
Shojaie, A., Basu, S., \& Michailidis, G. 2012.
\newblock Adaptive thresholding for reconstructing regulatory networks from
  time-course gene expression data.
\newblock {\em Statistics in Biosciences}, {\bf 4}(1), 66--83.

\bibitem[\protect\citename{Smith, }2012]{smith2012future}
Smith, Stephen~M. 2012.
\newblock The future of FMRI connectivity.
\newblock {\em Neuroimage}, {\bf 62}(2), 1257--1266.

\bibitem[\protect\citename{Tank {\em et~al.}, }2015]{tank2015bayesian}
Tank, Alex, Foti, Nicholas~J, \& Fox, Emily~B. 2015.
\newblock Bayesian structure learning for stationary time series.
\newblock {\em Pages  872--881 of:} {\em Proceedings of the Thirty-First
  Conference on Uncertainty in Artificial Intelligence}.
\newblock AUAI Press.

\bibitem[\protect\citename{Tsay, }1989]{Tsay_1989}
Tsay, Ruey~S. 1989.
\newblock Testing and modeling threshold autoregressive processes.
\newblock {\em Journal of the American Statistical Association}, {\bf 84}(405),
  231--240.

\bibitem[\protect\citename{Tseng \& Yun, }2009]{tseng2009coordinate}
Tseng, Paul, \& Yun, Sangwoon. 2009.
\newblock A coordinate gradient descent method for nonsmooth separable
  minimization.
\newblock {\em Mathematical Programming}, {\bf 117}(1), 387--423.

\bibitem[\protect\citename{Xiao \& Wu, }2012]{xiao2012covariance}
Xiao, Han, \& Wu, Wei~Biao. 2012.
\newblock Covariance matrix estimation for stationary time series.
\newblock {\em The Annals of Statistics}, {\bf 40}(1), 466--493.

\end{thebibliography}
\end{document}